\documentclass[12pt,a4paper]{article}
\usepackage{tikz}
\tikzset{>=latex}

\usepackage{amsfonts,amssymb,amsmath}

\usepackage{graphicx}
\usepackage{hyperref}
\usepackage{color}
\usepackage{float}
\usepackage{epstopdf}

\usepackage[normalem]{ulem} 

\graphicspath{{Mfiles/}}

\def\vect#1{\mbox{\boldmath{$#1$}}}

\def\Ac{{\cal A}}

\def\Nc{{\cal N}}
\def\Fc{{\cal F}}

\def\vx{\vec{\vect x}}
\def\vy{\vec{\vect y}}
\def\vz{{ \vec{\vect z}}}

\def\mC{\mathbb{C}}
\def\Cc{{\cal C}}
\newcommand{\cC}{\mathcal{C}}

\def\wG{G}

\newcommand{\bz}{\vect z}

\def\wg{g}

\newcommand{\argmin}{\mathop{\mbox{argmin}}}

\newcommand{\eps}{\varepsilon}
\renewcommand{\epsilon}{\eps}
\renewcommand{\leq}{\leqslant}
\renewcommand{\geq}{\geqslant}

\def\bfrho{\mbox{\boldmath$\rho$}}
\def\bfb{\mbox{\boldmath$b$}}
\def\bfeta{\mbox{\boldmath$\eta$}}
\def\bfd{\mbox{\boldmath$d$}}

\def\bfchi{\mbox{\boldmath$\chi$}}

\usepackage[normalem]{ulem}

{}

\newtheorem{theorem}{Theorem}

\newtheorem{remark}{Remark}
\newenvironment{proof}{\paragraph{{\bf Proof:}}}{\hfill$\square$}

\usepackage{perpage}
\MakePerPage{footnote}

\title{Fast signal recovery from quadratic measurements}

\author{Miguel Moscoso\footnote{Department of Mathematics, Universidad Carlos III de Madrid, Leganes, Madrid 28911, Spain}
, Alexei Novikov\footnote{Department of Mathematics, Pennsylvania State University, University Park, PA 16802}
, George Papanicolaou\footnote{Department of Mathematics, Stanford University, Stanford, CA 94305}
, Chrysoula Tsogka\footnote{Department of Applied Mathematics, University of California, Merced, CA 95343}
}

\begin{document}

\maketitle
\begin{abstract}
We present a novel approach for recovering a sparse signal from cross correlated data. The 
bottleneck for inversion in this case is the number of unknowns that grows quadratically, resulting 
in a prohibitive computational cost  as the dimension of the problem increases. 
The main feature of the proposed approach is that its cost is similar to the one of the usual  
sparse signal recovery problem with linear measurements. The keystone of the methodology is the use of a {\em Noise Collector}
that absorbs the data component that comes from the off-diagonal elements of the unknown matrix formed by the cross correlation of the original unknown vector. The data component that is absorbed this way does not carry extra information about the support of the signal and, thus, we can 
safely reduce the dimensionality of the problem making it competitive with respect to the one that uses linear data.
Our theory shows that the proposed approach provides exact support recovery when the data is not too noisy and that there are no false positives for any level of noise. Moreover, our theory also demonstrates that
 when using cross correlated data, the level of sparsity that can be recovered increases, scaling almost linearly with the number of data.
The numerical experiments presented in the paper
corroborate these findings.

\end{abstract}
\vspace{2pc}
\noindent{\it Keywords}: 
quadratic data, $\ell_1$-minimization, noise, model reduction


\maketitle
\section{Introduction}
Reconstruction of signals from 
cross correlations 
has interesting applications in many fields of science and engineering  such as optics, quantum mechanics, electron microscopy, antenna testing, seismic interferometry, or imaging in general \cite{Garnier16, Griffiths07, wap10, schuster}. 
Using cross correlations of measurements collected at different locations  
presents several advantages since the inversion does not require knowledge of the emitter positions, or the probing pulses shapes as only time differences matter.
Cross correlations have been used, for example,  when imaging is carried out with opportunistic sources whose properties are mainly unknown \cite{Garnier09,Garnier14,Daskalakis16,Helin18}.

In many applications, we seek information about an object or a signal $\bfrho\in\cC^{K}$ given  
data $\bfb\in\cC^{N}$ most often related through a linear transformation 
\begin{equation}
\label{eq:invprob}
\Ac \,\bfrho = \bfb \, ,
\end{equation}
where $\Ac\in\cC^{N\times K}$ is the measurement or model matrix. When the signal $\bfrho$ is compressed or when the data is scarce, $N < K$, in which case (\ref{eq:invprob}) is underdetermined and infinitely many signals or objects match the data. However, if the signal  $\bfrho$ is sparse so only  
$M\ll K$ components are different than zero, $\ell_1$-minimization algorithms that solve
\begin{equation}\label{l1normsol}
\bfrho_{\ell_1}=\argmin \|\vect \rho\|_{\ell_1}, \hbox{ subject to } \Ac \vect \rho= \vect b\,
\end{equation}
can recover the true signal efficiently even when $N\ll K$.

On the other hand, there are situations in which it is difficult or impossible to record high quality data, $\bfb$, and it is more convenient to use the cross correlated data contained in the matrix 
\begin{equation}
 \label{eq:cross0}
B =  \bfb\,\bfb^*\in\cC^{N\times N}\,  
\end{equation}
to find the desired information about the object or signal $\bfrho$ (see \cite{Garnier09} and references therein).
One way to address this problem is to lift it 
to the matrix level and reformulate it as a low-rank matrix linear system, 
which can be solved by using nuclear norm minimization as it was suggested in \cite{CMP11,Candes13} for imaging with intensities-only. This makes the problem convex over the appropriate matrix vector space and, thus, the unique true solution can be found using well established algorithms involving 
are not as efficient as $\ell_1$-minimization algorithms \textcolor{black}{that only involve lightweight operations such as matrix-vector multiplications \cite{Beck09}}. Furthermore, the big caveat 
is that the computational cost rapidly becomes prohibitively large as the dimension of the problem increases 
quadratically with $K$, making its solution infeasible.


In this paper we suggest a different approach. We propose to consider the linear matrix equation 
\begin{equation}
 \label{eq:modelC-intro}
 \Ac X \Ac^* = B\,  
\end{equation}
for the correlated signal $X=\bfrho\,\bfrho^*\in\cC^{K\times K}$,  vectorize both sides so
\begin{equation}
 \label{eq:vec00}
 \mbox{vec}(\Ac X \Ac^*) = \mbox{vec}(B)\, , 
\end{equation}
and use the Kronecker product $\otimes$, and its property $\mbox{vec}(PQR)=(R^T\otimes P)\mbox{vec}(Q)$, to express the matrix multiplications as the linear transformation
\begin{equation}
 \label{eq:vec1-intro}
 ({\bar \Ac} \otimes \Ac)\,\mbox{vec}( X ) = \mbox{vec}(B)\, .
\end{equation}
Thus, we can promote the sparsity of the sought image using $\ell_1$-minimization algorithms that are much faster than  nuclear norm minimization ones. 
However, the dimension of the unknown  $\mbox{vec}( X )$ in (\ref{eq:vec1-intro}) also increases quadratically with $K$, so this approach by itself would still be impractical when $K$ is not very small. 

Hence, we propose to use a {\em Noise Collector} to reduce the dimensionality of problem  (\ref{eq:vec1-intro}). The  {\em Noise Collector} was introduced in \cite{Moscoso20b} to eliminate the clutter in the recovered signals when the data are contaminated by additive noise. In this paper, we use the Noise Collector to absorb part of the signal instead. Specifically, we treat as noise the signal 
 that corresponds to the $K^2 -K$ off-diagonal entries in the matrix $X$. Using the {\em Noise Collector} allows us to ignore these 
entries and construct a linear system with the same number of unknowns as the original problem (\ref{eq:invprob}) that uses linear data. 
As a consequence a dimension reduction from $K^2$ to $K$ unknowns is achieved. 
The main result of this paper is Theorem \ref{d-reduce} which says that under certain decoherence conditions on the matrix $\Ac$, we can find the support of an M-sparse signal exactly if the data is noise-free or the noise is low enough.  Furthermore, Theorem \ref{d-reduce} shows that the level of sparsity $M$ that can be recovered increases from  $O(\sqrt{N}/\sqrt{\ln N})$ to $O(N/\sqrt{\ln N})$ when quadratic cross correlation data are used instead of the linear ones.

The numerical experiments included in this paper support the results of Theorem \ref{d-reduce}. They show that the support of a signal can be found exactly if the noise in the data is not too large with almost no extra computational cost with respect to the original  problem (\ref{eq:invprob}) that considers linear data with no correlations. Once the support has been found, a trivial second step allows us to find the signal, including its phases. The reconstruction is exact when there is no noise in the data and the results are very satisfactory even for noisy data with low signal to noise ratios. 
That is, our numerical experiments suggest that the approach presented here is robust with respect to additive noise.
Additional properties of this approach are that for any level of noise the solution has no false positives, and that the algorithm is parameter-free, so it does not require an estimation of the energy of the {\em off-diagonal signal} that we need to absorb, or of the level of noise in the data.

The paper is organized as follows. In Section \ref{sec:passive}, we summarize the model used 
to generate the signals to be recovered, which in our case are images. In Section \ref{sec:NC}, we present the theory that supports the proposed strategy for dimension reduction  when correlated data are used to recover the signals. Section \ref{sec:algo} explains the algorithm for carrying out the inversion efficiently. Section \ref{sec:numerics} shows the numerical experiments. Section \ref{sec:conclusions}
summarizes our conclusions. The proofs of the theorems are given in \ref{sec:proofs}. 

\section{Passive array imaging}
\label{sec:passive}
We consider processing of passive array signals where the object to be imaged is a set of point sources at positions $\vz_{j}$ and  
(complex)
amplitudes $\alpha_j$,
$j=1,\dots,M$.  The data used to  image the object are collected at several sensors on an array; see Figure \ref{fig:setup}. 
The imaging system is characterized by the array aperture $a$, the distance $L$ to the sources, the bandwidth $B$ 
and the central wavelength $\lambda_0$ of the signals.      
\begin{figure}[htbp]
\includegraphics[scale=1]{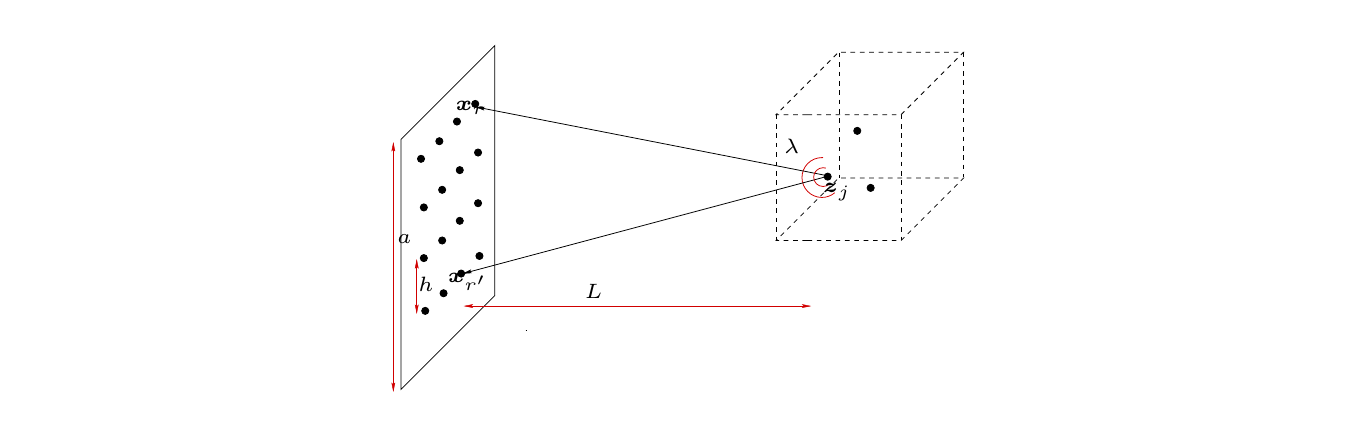}
\caption{General setup for passive array imaging. The source at $\vz_j$  emits a signal that is recorded at all array elements $\vx_r$, $r=1,\ldots,\textcolor{black}{N_r}$.}
\label{fig:setup}
\end{figure}

The sources are located inside an image window IW discretized with a uniform grid of points $\vy_k$, $k=1,\ldots,K$. Thus,
the signal to be recovered 
is the source vector 
\begin{equation}
\label{eq:rhotilde}
\vect {\tilde \rho}=[{\tilde \rho}_{1},\ldots,{\tilde \rho}_{K}]^{\intercal}\in\mathbb{C}^K\,,
\end{equation}
whose components ${\tilde \rho}_k$ correspond to the  amplitudes of the $M$ sources at the grid points $\vy_k$, $k=1,\ldots,K$, with $K\gg M$. This vector has components
${\tilde\rho}_{k} = \alpha_j$ if  $\vy_{k} = \vz_j$ for some  $j=1,\ldots,M$, while the others  are zero. 

Denoting by $\wG(\vx,\vy;\omega)$ the Green's function for the propagation of a 
wave of angular frequency $\omega$ from point $\vy$ to point $\vx$, we define the single-frequency 
Green's function vector that connects a point $\vy$ in the IW with all the sensors on the array located at points 
$\vx_r$, $r = 1,\ldots,\textcolor{black}{N_r}$, so
$$
\vect \wg(\vy;\omega)=[\wG(\vx_{1},\vy;\omega), \wG(\vx_{2},\vy;\omega),\ldots,
\wG(\vx_{N},\vy;\omega)]^{\intercal} \in \mathbb{C}^{\textcolor{black}{N_r}} \,.
$$
In three dimensions,
$
{\displaystyle
\wG(\vx,\vy;\omega)=  \frac{\exp\{i \omega|\vx-\vy|/c_0\}}{4\pi|\vx-\vy|} } 
$ 
if the medium is homogeneous. Hence, the  signals of frequencies $\omega_l$ recorded at the sensors locations $\vx_r$ are
$$b(\vx_r, \omega_l)=\sum_{j=1}^M\alpha_j  \wG (\vx_r,\vz_j; \omega_l)\, ,\quad r = 1,\ldots,\textcolor{black}{N_r}\, .
$$
They form the single-frequency data vector $\vect b(\omega_l) =  [ b(\vx_1, \omega_l), b(\vx_2, \omega_l),\dots,b(\vx_N, \omega_l)]^{\intercal}\in \mathbb{C}^{\textcolor{black}{N_r}}$.
As several frequencies $\omega_l$, $l=1,\dots,N_f$, are used to recover (\ref{eq:rhotilde}), all the recorded data are stacked in the multi-frequency  column data vector
\begin{equation}
\vect b = [ \vect  b(\omega_1)^{\intercal},\vect  b(\omega_2)^{\intercal},\dots, \vect b(\omega_{N_f})^{\intercal}]^{\intercal} 
\in \mathbb{C}^{\textcolor{black}{N}} \, , \mbox{with} \, \textcolor{black}{N=N_r N_f}\,.
\label{eq:data}
\end{equation}

\subsection{The inverse problem with linear data}
When the data (\ref{eq:data}) are available and reliable, 
one can form the linear system
\begin{equation}
\Ac\,\bfrho = \bfb
\label{eq:system}
\end{equation}
to recover (\ref{eq:rhotilde}). Here, $ \Ac$ is the $\textcolor{black}{N}\times K$ measurement matrix whose columns $\vect a_k$ are the multi-frequency Green's function vectors
\begin{equation}\label{eq:ak}
\vect a_k   = \frac{1}{c_k}\, 
[ \vect\wg(\vy_k ; \omega_1)^{\intercal},  \vect\wg(\vy_k; \omega_2)^{\intercal}, \dots, \vect\wg(\vy_k; \omega_{N_f})^{\intercal} ]^{\intercal} \in \mathbb{C}^{\textcolor{black}{N}}\, ,
 \end{equation}
where $c_k$ are scalars that normalize these vectors to have $\ell_2$-norm one, and 
\begin{equation}\label{eq:rho}
\bfrho=\mbox{diag}(c_1, c_2,\dots,c_K)\,\vect {\tilde \rho},
 \end{equation}
 where $\vect {\tilde \rho}$ is given by (\ref{eq:rhotilde}).
Then, one can solve (\ref{eq:system}) for the unknown vector $\bfrho$ using a number of $\ell_2$ and $\ell_1$ inversion methods to find the sought image. 
In general,  $\ell_2$  methods are robust but the resulting resolution  
is low. On the other hand,  $\ell_1$ methods provide higher resolution but they are much more sensitive to noise in the data. Hence, they cannot be used with poor quality data unless one  carefully takes care of the noise.

\subsection{The inverse problem with quadratic cross correlation data}
In many instances, imaging with cross correlations helps to form better and more robust images. 
This is the case, for example, when one uses high frequency signals and  has a  low-budget measurement system with inexpensive sensors that
are not able to resolve the signals well. Another situation is when the raw data (\ref{eq:data}) can be measured but it is more convenient to image with cross correlations because they help to mitigate the effects of the inhomogeneities of the  medium between the sources and the sensors \cite{bakulin,Garnier15}

Assume that all the cross correlated data  contained in the matrix
\begin{equation}
B =  \bfb\,\bfb^*\in\cC^{N\times N}\, 
 \end{equation}
 are  available for imaging. Then, one can consider the linear system 
\begin{equation}
 \label{eq:modelC}
 \Ac X \Ac^* = B\, ,
 \end{equation}
and seek the correlated image $X=\bfrho\,\bfrho^*\in\cC^{K\times K}$ that solves it. The unknown matrix $X$ is  rank 1 and, hence, one possibility is to look for a low-rank matrix by using nuclear norm minimization as it was suggested for imaging with intensities-only in \cite{CMP11,Candes13}. This is possible in theory, but it is unfeasible when the problem is large because the number of unknowns grows quadratically and, therefore, the computational cost rapidly becomes prohibitive. For example, to form an image with $1000\times 1000$ pixels one would have to solve a system with $10^{12}$ unknowns.

Instead, we suggest the following strategy. We propose to vectorize both sides of (\ref{eq:modelC}) so
\begin{equation}
 \label{eq:vec0}
 \mbox{vec}(\Ac X \Ac^*) = \mbox{vec}(B)\, , 
\end{equation}
where $\mbox{vec}(\cdot)$ denotes the vectorization of a matrix formed by stacking its columns into a single column vector. Then, we
use the Kronecker product $\otimes$, and its property $\mbox{vec}(PQR)=(R^T\otimes P)\mbox{vec}(Q)$, to express the matrix multiplications as the linear transformation
\begin{equation}
 \label{eq:vec1}
 ({\bar \Ac} \otimes \Ac)\,\mbox{vec}( X ) = \mbox{vec}(B)\, .
\end{equation}
With this formulation of the problem we can use an $\ell_1$ minimization algorithm to form the images, which is much faster
than a nuclear norm minimization algorithm that needs to compute the SVD  of the 
iterate  matrices. 
However, with just this approach the main obstacle is not overcome, as the dimensionality 
still grows quadratically with the number of unknowns $K$. Hence, we propose  a dimension reduction strategy that uses a {\em Noise Collector} \cite{Moscoso20b} to absorb a component of the data vector that does not provide extra information about the signal support. We point out that this component is not a gaussian random vector as in \cite{Moscoso20b}, but a deterministic vector resulting from the off-diagonal terms of $X$ that are neglected.

\section{The noise collector and  dimension reduction }
\label{sec:NC}


\subsection{The Noise Collector}
\label{subsec:PlainNC}

The {\em Noise Collector}~\cite{Moscoso20b} is a method to find  the vector $ {\mbox{\boldmath$\chi$}}  \in \mathbb{C}^{\cal K}$ in
\begin{equation}
\label{eq:nc_general}
 T \, \vect \chi = \vect d_0 + \vect e \,,
\end{equation}
from highly incomplete measurement data $\vect d = \vect d_0 + \vect e \in \mC^{\cal N}$ possibly
corrupted by  noise  $\vect e \in \mC^{\cal N}$, where $ 1 \ll {\cal N} < {\cal K}$.  Here, $T$ is a general measurement matrix of size
${\cal N}\times {\cal K}$, whose columns have unit length. The main results in~\cite{Moscoso20b} ensure that we can still
recover the support of  $\vect \chi$ when the data is noisy by looking at the support of  $\vect \chi_{\tau}$ found as 
\begin{equation}\label{rho_tt}
\left( \vect \chi_{\tau}, \vect \eta_{\tau} \right) = 
\arg\min_{ \small \vect \chi, \small \vect \eta} \left(    \tau  \| \vect \chi \|_{\ell_1} +  \| \vect \eta \|_{\ell_1}  \right),
 \hbox{ subject to } T \vect \chi + \Cc \vect \eta =\vect d, 
\end{equation}
with an {\color{black} $O(1)$} no-phatom weight $\tau$, and a {\it Noise Collector} matrix  $\cC \in \mC^{{\cal N}\times \Sigma}$ with $\Sigma = {\cal N}^\beta$, for  $\beta>1$.
If  the noise $\vect e$ is Gaussian, then
the columns of $\Cc$ can be chosen independently and at random on the unit sphere
 $\mathbb{S}^{{\cal N}-1}=\left\{ x \in \mathbb{R}^{{\cal N}} , \| x \|_{\ell_2} =1  \right\}$. The weight $\tau>1$ is chosen so it is expensive to approximate $\vect e$ with the columns
of $T$, but it cannot be taken too large because then we loose the signal $\vect \chi$ that gets absorbed by the {\it Noise Collector} as well. Intuitively,  $\tau$
is a measure of the rate at which the signal is lost as the noise increases. For practical purposes, $\tau$ is chosen as the minimal value for which $\vect \chi = 0$
when the data is pure noise, i.e., when $\vect d_0=0$. The key property is that the optimal value of $\tau$ does not depend on the level of
noise and, therefore, it is chosen in advance, before the {\it Noise Collector} is used for a specific task. We have the following result.

\begin{theorem}\label{pnas}~\cite{Moscoso20b} 
Fix $\beta>1$,  and draw  $\Sigma={\cal N}^{\beta}$  columns to form the Noise Collector  
$\Cc$, independently, from the  uniform distribution on 
$\mathbb{S}^{{\cal N}-1}$.  Let $\vect \chi$ be an $M$-sparse solution of the noiseless system $T \vect \chi=\vect d_0$,   and $\vect \chi_\tau$  the solution 
of~(\ref{rho_tt}) with $ \vect d= \vect d_0 + \vect e$.   Denote the ratio of minimum to maximum  significant values of $\vect \chi$ as
 \begin{equation}
 \label{def:gamma}
 \gamma= \min_{i \in \mbox{supp}(\vect \chi)} \frac{|\chi_i|}{ \| \vect \chi \|_{\ell_\infty}}.
 \end{equation}
 Assume that the columns of $T$  are incoherent,  so that 
\begin{equation}\label{Mcond}
|\langle \vect t_i, \vect t_j \rangle| \leq \frac{1}{3M} \mbox{ for all } i \mbox{ and } j.
\end{equation}
Then, for any $\kappa > 0$, there are constants  $\tau=\tau(\kappa, \beta)$, $c_1=c_1(\kappa, \beta, \gamma)$, and ${\cal N}_0= {\cal N}_0(\kappa, \beta)$ such that,
if the noise level satisfies   
\begin{equation}\label{esti}
\max\left(1,  \| \vect e \|_{\ell_2}\right)
 \leq  c_1  \frac{\| \vect d_0\|_{\ell_2}^2 }{ \| \vect \chi \|_{\ell_1}}  \sqrt{\frac{{\cal N}}{ \ln {\cal N}}},
\end{equation}
 then
$\mbox{supp}(\vect \chi_{\tau}) = \mbox{supp}(\vect \chi)$ for  all ${\cal N}>{\cal N}_0$ 
with  probability  $1-1/{\cal N}^{\kappa}$. 
\end{theorem}

To gain a better understanding of this theorem, let us consider the case where $T$ is the identity matrix (the classical denoising problem) and all coefficients of $\vect d_0= \vect \chi$ 
are either 1 or 0. Then  $\| \vect d_0 \|^2_{\ell_2}  =   \| \vect \chi \|_{\ell_1}=M$. In this case, an acceptable level of noise is
  \begin{equation}\label{thm3}
  \| \vect e \|_{\ell_2} \lesssim   \| \vect d_0 \|_{\ell_2}  \sqrt{ \frac{\cal N} {M \ln {\cal N}}} \sim  \sqrt{ \frac{\cal N} {\ln {\cal N}}}.
  \end{equation}
The estimate~(\ref{thm3}) implies that we can handle more noise as we increase the number of measurements. This holds for two reasons. Firstly,  a typical noise vector $\vect e$  is almost orthogonal 
to the columns of $T$, so
\begin{equation}\label{no-random}
 |\langle \vect t_i,\vect e \rangle | \leq c_0\sqrt{ \frac {\ln {\cal N}}{\cal N}} \| \vect e \|_{\ell_2}
\end{equation}
for some $c_0= c_0 (\kappa)$ with probability  $1-1/{\cal N}^{\kappa}$.
In particular, a typical noise vector $\vect e$  is almost orthogonal 
to the signal subspace $V$. More formally, suppose  $V$ is the $M$-dimensional subspace spanned by the column vectors $\vect t_j$ with $j$ in the support of $\vect \chi$, and let $W=V^{\perp}$ be the orthogonal complement to $V$. Consider the orthogonal decomposition  $\vect e =\vect e^{v} + \vect e^{w}$, such that $\vect e^{v}$ is in $V$ 
 and $\vect e^{w}$ is in $W$. 
 Then,
 \[
 \| \vect e^{v} \|_{\ell_2}   \lesssim  \sqrt{\frac M{\cal N}} \| \vect e \|_{\ell_2}
 \]
 with high probability that tends to $1$, as ${\cal N} \to \infty$. In Theorem~\ref{pnas}, a quantitative estimate of this convergence is $1-1/{\cal N}^{\kappa}$.
It means that if a signal is sparse so $M \ll {\cal N}$, then  we can recover it for very low signal-to-noise ratios.
Secondly, and more importantly, if the columns of the noise collector $\cC$ are also almost orthogonal to the signal subspace, then it is 
too expensive to approximate the signal $\vect d_0$ with the columns of $\cC$ and, hence, we have to use the columns of the measurement matrix $T$. 
If we draw the columns of $\cC$, independently, from the  uniform distribution on $\mathbb{S}^{{\cal N}-1}$, then they will be almost orthogonal to the signal subspace with high probability. It is again
estimated as $1-1/{\cal N}^{\kappa}$ in Theorem~\ref{pnas}. 
Finally, the incoherence condition~(\ref{Mcond}) implies that it is too expensive to approximate the signal $\vect d_0$ with columns $T$ that are not in the support of $\vect \chi$
and, hence, there are no false positives. 

In Theorem~\ref{pnas} we used randomness twice: the noise vector  $\vect e$ was random and the columns of the noise collector were drawn at random. Note that in both  cases randomness could be 
replaced by deterministic conditions requiring that  $\vect e$ and the columns of $\cC$  are almost orthogonal to the  signal  subspace. It is natural to assume that the noise vector  $\vect e$ is a random variable and, as we explain in~\cite{Moscoso20b},
 the columns of $\cC$ are random 
because it is hard to construct a deterministic $\cC$ that satisfies the almost orthogonality conditions. In the present work
we still construct the matrix  $\cC$ randomly, but we sometimes treat the vector $\vect e$ as deterministic, as for example, in our Theorem~\ref{thm-ortho}. Inspection of the proofs in~\cite{Moscoso20b} shows that the only condition
on $\vect e$ we need to verify from Theorem~\ref{pnas} is~(\ref{no-random}). Thus, the next Theorem is a deterministic reformulation of Theorem~\ref{pnas}. The proof is given in  \ref{sec:proof2}.

\begin{theorem}\label{thm-ortho}
Assume conditions on $\vect \chi$, $T$, and $\Cc$ 
are as in Theorem~\ref{pnas} and  define $\gamma$ as in (\ref{def:gamma}).
Then, for any $\kappa > 0$, there are constants  $\tau_0=\tau_0(\kappa, \beta)$,  $c_0=c_0(\kappa, \beta)$,   and ${\cal N}_0= {\cal N}_0(\kappa, \beta, \gamma)$,
 $\alpha=\alpha(c_0, \kappa, \beta)$
such that the following two claims hold.

(i) If $\vect e$ satisfies~(\ref{no-random})
 for all $\vect t_i$, $i \not\in \mbox{ supp }(\vect \chi) $; all columns of $T$ satisfy
\begin{equation}\label{no-random-2}
 |\langle \vect t_i,\vect t_j \rangle | \leq c_0 \frac{\sqrt{\ln {\cal N}}}{\sqrt{\cal N}} 
\end{equation}
for all $i$ and $j$; the sparsity $M$ is such that 
\begin{equation}\label{eq:M}
  M \leq \alpha \frac{ \sqrt{\cal N}}{\sqrt{\ln {\cal N}}};
\end{equation}
  and  $\tau \geq  \tau_0$, 
then  $\mbox{supp}(\vect \chi_{\tau})  \subset \mbox{supp}(\vect \chi)$ 
  with  probability  $1-1/{\cal N}^{\kappa}$.
  
(ii) If, in addition,  the noise is not large, so
\begin{equation}\label{new-estimate1}
\left| \langle \vect t_m,  \vect e \rangle \right| \leq   \min_{i \in \mbox{supp}(\vect \chi)} |\chi_i|/2
 \end{equation}
  for all $\vect t_m$, $m \in \mbox{supp}(\vect \chi)$,
and 
\begin{equation}\label{new-estimate2}
  \| \vect e\|_{\ell_2} \leq  c_1  \| \vect \chi \|_{\ell_1}
\end{equation}
for some $c_1$, then
$\mbox{supp}(\vect \chi) = \mbox{supp}(\vect \chi_{\tau})$ for  all ${\cal N}>{\cal N}_0$ 
with  probability  $1-1/{\cal N}^{\kappa}$. 
\end{theorem}

%

 In contrast to Theorem~\ref{pnas}, we require  in Theorem~\ref{thm-ortho} condition~(\ref{no-random}) to hold only for   
 for  $\vect t_i$, $i \not\in \mbox{ supp }(\vect \chi)$, that is for the columns of $T$ outside the support of $\vect \chi$.
 For the columns inside the support, $i \in \mbox{ supp }(\vect \chi)$, we relax condition~(\ref{no-random}) to condition~(\ref{new-estimate1}). Thus 
 Theorem~\ref{thm-ortho} has slightly weaker assumptions than Theorem~\ref{pnas}. For a random $\vect e$ this weakening in not essential, because one needs  to 
 know the support  of $\vect \chi$ in advance. It turns out that for our $\vect e$  this weakening will become important (see Remark~\ref{new_e} in the end of~\ref{sec:proof3}) .

\subsection{Dimension reduction for quadratic cross correlation data }
\label{subsec:dimReduce}

The $N^2\times K^2$  linear problem (\ref{eq:vec1}) that uses quadratic cross correlation data is notoriously hard to solve due to its high dimensionality. Therefore, 
 we propose the following strategy for robust dimensionality reduction.  The idea is to treat the contribution of the off-diagonal elements of $X=\bfrho\,\bfrho^*\in\cC^{K\times K}$ 
 as {\it noise} and, thus, use the {\it Noise Collector} to absorb it. Namely, we define
\begin{equation}
\label{eq:def0}
\vect \chi =\mbox{diag}(X) = [|\rho_1|^2, |\rho_2|^2, \dots, |\rho_K|^2]^T\, ,
\end{equation}
and re-write~(\ref{eq:vec1}) as
\begin{equation}
\label{eq:dimred}
 T \, {\mbox{\boldmath$\chi$}} + \cC \, \bfeta= \bfd\, ,
\end{equation}
where  {\color{black} we replace the off-diagonal elements by the Noise Collector term $\cC \, \bfeta$ and}
\begin{equation}
\label{eq:def1}
T=({\bar \Ac} \otimes \Ac)_{\vect \chi}
\end{equation} 
contains only the $K$ columns of  ${\bar \Ac} \otimes \Ac$ corresponding to ${\mbox{\boldmath$\chi$}}$.  Thus,  the size of  $\vect \chi$ is ${\cal K}$ and the size of $T$ is ${\cal N} \times {\cal K}$, 
with ${\cal K}=K$ and ${\cal N} =N^2$. In practice, the measurements may be subsampled as well, so the size of the system can be further reduced to ${\cal N} \times {\cal K}$, with ${\cal N} = O(N)$ and  ${\cal K}=K$.

Problem (\ref{eq:dimred})
can be understood as an exact linearization of the classical phase retrieval problem, where all the interference terms $\rho_i \rho^*_j$ for $i\ne j$ are absorbed 
in $\cC \, \bfeta$, with $\bfeta$ being an unwanted vector considered to be noise in this formulation. 
In other words, the phase retrieval problem with $K$ unknowns has been transformed to the linear problem (\ref{eq:dimred}) that also has  $K$ unknowns.  
Note, though, that  in phase retrieval only autocorrelation measurements are considered, while in (\ref{eq:dimred}) we also use cross-correlation measurements.

In the next theorem  we use all the  measurements  $\vect d \in \mathbb{C}^{\cal N}$,   so ${\cal N} =N^2$   in~(\ref{eq:dimred}). This is done for simplicity of presentation, but in practice  ${\cal N} =O(N)$ measurements  are enough.
We will choose a solution of~(\ref{eq:dimred})
 using~(\ref{rho_tt}). As in Theorems~\ref{pnas} and \ref{thm-ortho},  the vector $\bfeta$ in (\ref{eq:dimred}) has ${\cal N}^{\beta}$ entries that do not have physical meaning. Its only purpose is to absorb 
the off-diagonal contributions in $\vect e = \bfd - T \vect \chi $. We point out that the magnitude of $\vect e$ is not small if $M\ge 2$. Indeed, 
the contribution of ${\mbox{\boldmath$\chi$}}=\mbox{diag}(X)$ to the data $\bfd$  is of order $M$, while the contribution of the off-diagonal terms of $X$ is of order $M^2$. Furthermore,  the
vector  $\vect e$ is not independent of $\vect \chi$ anymore.

\begin{theorem}\label{d-reduce} 
Fix  $|\rho_i|$. Suppose  the phases $\rho_i/|\rho_i|$  are independent  and uniformly distributed on the (complex) unit circle.  
Suppose $X$ is a solution of~(\ref{eq:vec1}), $ \vect \chi =\mbox{diag}(X)$ is $M$-sparse, 
and
$T=({\bar \Ac} \otimes \Ac)_{\vect \chi}: \mathbb{C}^{\cal K} \to \mathbb{C}^{\cal N}$,   ${\cal K}=K$ and ${\cal N} = N^2$.
Fix $\beta>1$,  and draw  $\Sigma={\cal N}^{\beta}$  columns for 
$\Cc$, independently, from the  uniform distribution on 
$\mathbb{S}^{{\cal N}-1}$. 
 Denote 
 \begin{equation}
 \label{def:Delta}
 \Delta = \sqrt{N}  \max_{i \neq j} |\langle \vect a_i, \vect a_j \rangle|,
 \end{equation}
 and define $\gamma$ as in (\ref{def:gamma}). Then, for any $\kappa > 0$, there are constants $\alpha=\alpha(\kappa, \gamma, \Delta)$, $\tau=\tau(\kappa, \beta)$, and ${\cal N}_0= {\cal N}_0(\kappa, \beta, \gamma, \Delta)$ such that
the following holds. If  
$M \leq \alpha N/\sqrt{\ln N}$ 
and $ \vect \chi_{\tau}$  is the solution of~(\ref{rho_tt}),
 then
$ \mbox{supp}(\vect \chi) = \mbox{supp}(\vect \chi_{\tau})$ for  all ${\cal N}>{\cal N}_0$ 
with  probability  $1-1/{\cal N}^{\kappa}$. 
\end{theorem}

{\color{black} The proof of Theorem~\ref{d-reduce} is given in \ref{sec:proof3}. 
In Theorem~\ref{d-reduce}  the scaling for sparse recovery is $M \leq \alpha N/\sqrt{\ln N}$. This result is in good agreement with our numerical experiments, see Figure~\ref{fig_phase}. In order   
to obtain this scaling we introduced our probabilistic framework - in Theorem~\ref{d-reduce}  assuming that the phases of the signals  are random. The idea is that a vector with random phases better  describes a typical  signal in many applications.
The dimension reduction, however, could be done without introducing the  probabilistic framework. We state 
and prove a deterministic version of Theorem~\ref{d-reduce} in~\ref{sec:proof4}  for completeness.
  In this case the scaling for sparse recovery is more conservative:   $M \leq \alpha \sqrt{N}/\sqrt{\ln N}$, and it does not agree with our numerical experiments. 
}

\section{Algorithmic implementation}
\label{sec:algo}

A key point of the propose strategy is that the $M$-sparsest solution of (\ref{eq:dimred}) can be effectively found  by solving the minimization problem
\begin{eqnarray}\label{rho_t}
\left( \bfchi_{\tau}, \vect \eta_{\tau} \right) = 
\arg\min_{ \small \vect {\mbox{\boldmath$\chi$}} , \small \vect \eta} \left(    \tau  \| \vect {\mbox{\boldmath$\chi$}}  \|_{\ell_1} +  \| \vect \eta \|_{\ell_1}  \right),\\
 \hbox{ subject to }T \, {\mbox{\boldmath$\chi$}} + \Cc \vect \eta =\bfd  , \nonumber 
\end{eqnarray}
with an $O(1)$ no-phatom weight $\tau$.  
The main property of this approach is that if the matrix $T$ is incoherent enough, so its columns satisfy assumption (\ref{Mcond}) of Theorem \ref{pnas},  
the $\ell_1$-norm minimal solution of (\ref{rho_t}) has a zero false discovery rate for any level of noise, with probability that tends to one as the dimension of the data $\Nc$ increases to infinity. More specifically, the relative level of noise that the {\em Noise Collector} can handle is of order $O(\sqrt{\Nc}/\sqrt{M \ln \Nc})$.  Below this level of noise there are no false discoveries. 


To find the minimizer in  (\ref{rho_t}), we  define the function
\begin{eqnarray} 
 F(\bfchi, \vect \eta, \bz) &=& \lambda\,(\tau \| \bfchi \|_{\ell_1} +  \| \vect \eta \|_{\ell_1}) \\
&+& \frac{1}{2} \| T  \vect \chi  + \Cc  \vect \eta - \bfd \|^2_{\ell_2} + \langle \bz, \bfd - T \bfchi - \Cc  \vect \eta \rangle \nonumber
\end{eqnarray}
for a no-phantom weight $\tau$, and determine the solution as
\vspace{-0.2cm}
\begin{equation}\label{min-max}
\max_{\bz} \min_{\bfchi,\vect \eta} F(\bfchi,\vect \eta,\bz) .
\end{equation}
This strategy finds the minimum in (\ref{rho_t}) exactly for all values of the regularization parameter $\lambda$. Thus,
the  method is fully automated, meaning that it 
has no tuning parameters. To determine the exact extremum in (\ref{min-max}), we use the iterative soft thresholding algorithm GeLMA~\cite{Moscoso12}
 that works as follows.
 
Pick a value for the no-phantom weight $\tau$; for optimal results calibrate $\tau$ to be the smallest value for which $\vect \chi=0$ when the algorithm is fed with pure noise.  
In our numerical experiments we use $\tau= 2$.
Next, pick a value for the regularization parameter, for example $\lambda=1$, and choose step sizes $\Delta t_1< 2/\|[T \, | \, \Cc]\|^2$ and 
$\Delta t_2< \lambda/\|T\|$\footnote{Choosing two step sizes instead of the smaller one $\Delta t_1$ improves the convergence speed.}. Set $\vect \bfchi_0= \vect 0$, $\vect \eta_0=\vect 0$, $\vect z_0=\vect 0$, and
iterate for $k\geq 0$:
\begin{eqnarray} 
&& \vect r = \bfd - T \,\bfchi_k - \Cc \,\vect\eta_k\nonumber \, ,\\
&&\vect \bfchi_{k+1}=\mathcal{S}_{ \, \tau \, \lambda \Delta t_1} ( \bfchi_k +\Delta t_1 \, T^*(\vect z_k+ \vect r))
\nonumber \, ,\\
&&\vect \eta_{k+1}=\mathcal{S}_{\lambda \Delta t_1} ( \vect\eta_k +\Delta t_1  \, \Cc^*(\vect z_k+ \vect r))
\nonumber \, ,\\
&&\vect z_{k+1} = \vect z_k + \Delta t_2 \, \vect r \label{eq:algo}\, ,
\end{eqnarray}
where  $\mathcal{S}_{r}(y_i)=\mbox{sign}(y_i)\max\{0,|y_i|-r\}$. 
\subsection{The Noise Collector: construction and properties} 
To construct the {\em Noise Collector} matrix $\cC \in \mC^{\Nc \times \Nc^\beta}$ that satisfies the assumptions of Theorem  \ref{pnas} one could draw $\Nc^\beta$ normally distributed $\Nc$-dimensional vectors, normalized to unit length. Thus, the additional computational cost incurred for implementing the {\em Noise Collector}  in (\ref{eq:algo}), due to the terms $\Cc \vect  \eta_k$ and $\Cc^* (\vect z_k + \vect r)$, would be $O(\Nc^{\beta+1})$, which is not very large as we use $\beta \approx 1.5$  in practice. The computational cost of (\ref{eq:algo})  without the {\em Noise Collector} mainly comes from the matrix vector multiplications $T \,\bfchi_k$ which can be done
in $O(\Nc{\cal K})$ operations and, typically, ${\cal K} \gg \Nc$.

To further reduce the additional computational time and memory requirements we use a different construction procedure that exploits the properties of circulant matrices. The idea is to  
draw instead a few normally distributed $\Nc$-dimensional vectors of length one, and construct from each one of them a circulant matrix of dimension $\Nc \times \Nc$. The columns of these matrices are still independent and uniformly distributed on $\mathbb{S}^{\Nc-1}$, so  they satisfy the assumptions of Theorem \ref{pnas}. The full {\em Noise Collector} matrix is then formed by concatenating these circulant matrices together. 

More precisely, the {\em Noise Collector} construction is done in the following way. We draw $\Nc^{\beta-1}$ normally distributed $\Nc$-dimensional vectors, normalized to unit length. 
These are the generating vectors of the {\em Noise Collector}. To these vectors are associated $\Nc^{\beta-1}$ circulant matrices $\cC_i \in \mC^{\Nc \times \Nc}$, $i=1, \ldots, \Nc^{\beta-1}$, and the {\em Noise Collector} matrix is constructed by concatenation of these $\Nc^{\beta-1}$ matrices, 
so 
$$
\Cc = \left[ \Cc_1 \left| \Cc_2 \left|  \cC_3 \left| \ldots  \right. \right. \right. \left| \cC_{\Nc^{\beta-1}} \right. \right] \in \mC^{\Nc \times \Nc^\beta}.
$$ 
We point out that the {\em Noise Collector} matrix $\Cc$ is not stored, only the $\Nc^{\beta-1}$ generating vectors are saved in memory. On the other hand, the matrix vector multiplications $\Cc \vect  \eta_k$ and $\Cc^* (\vect z_k + \vect r)$ in (\ref{eq:algo}) can be computed using these generating vectors and FFTs \cite{Gray06}. This makes the complexity associated to the {\em Noise Collector}  $O(\Nc^{\beta} \log(\Nc))$. 


To explain this further, we recall briefly below how a matrix vector multiplication can be performed using the FFT for a circulant matrix. 
For a generating vector $\vect c= [c_0, c_1,\ldots,c_{\Nc-1}]$, the $\cC_i$ circulant matrix takes the form
$$\cC_i = \left[ 
\begin{array}{llll}
c_0 & c_{\Nc-1} & \ldots & c_1 \\
c_1 & c_{0} & \ldots & c_2 \\
\vdots & & \ddots & \vdots \\
c_{\Nc -1}  & c_{\Nc -2} & \ldots & c_0 \\
\end{array}
\right] \, .
$$
This matrix can be diagonalized by the Discrete Fourier Transform (DFT) matrix, i.e., 
 $$ \cC_i =  \Fc \Lambda \Fc^{-1} $$ 
 where $\Fc$ is the DFT matrix, $\Fc^{-1}$ is its inverse, and $\Lambda$ is a diagonal matrix such that $\Lambda = \mbox{diag}(\Fc \vect c)$, where $\vect c$ is the generating vector. 
 Thus, a matrix vector multiplication $\cC_i \vect \eta$ is performed as follows: (i) compute $\vect {\hat \eta} = \Fc^{-1} \vect \eta$, the inverse DFT of $\vect \eta$ in $\Nc \log(\Nc)$ operations, (ii) compute the eigenvalues of $\cC_i$ as the DFT of $\vect c$, and component wise multiply the result with $\vect {\hat \eta}$ (this step can also be done in $\Nc \log(\Nc)$ operations), and  (iii)  compute the FFT of the vector resulting from step (ii) in, again,  $\Nc \log(\Nc)$ operations. 
 
 Consequently, the cost of performing the multiplication $\Cc \vect  \eta_k$ is $\Nc^{\beta-1} \Nc \log(\Nc)= \Nc^{\beta}  \log(\Nc)$. As the cost of finding the solution without the {\em Noise Collector} is $O(\Nc {\cal K})$ due to the terms $T \,\bfchi_k$, the additional cost due to the {\em Noise Collector} is negligible since ${\cal K} \gg \Nc^{\beta-1}\log(\Nc)$ because, typically, ${\cal K} \gg \Nc$ and $\beta \approx 1.5$.  

\section{Numerical results}
\label{sec:numerics}
We consider processing of passive array signals. We seek to determine the positions $\vz_{j}$ and the 
complex amplitudes $\alpha_j$ of $M$ point sources, $j=1,\dots,M$, from measurements of polychromatic signals on an array of receivers; see Figure \ref{fig:setup}. 
The source imaging problem is considered here for simplicity. The active array imaging problem can be cast under the same linear algebra framework even when multiple scattering is important \cite{CMP14}.

The array consists of $N_r=21$ receivers located at  $x_r=-\frac{a}{2}+\frac{r-1}{N_r-1} a$, $r=1,\ldots,N_r$, where $a=100\lambda$ is the array aperture. The imaging window (IW) is at range $L=100 \lambda$ from the array and the bandwidth $B=f_0/3$ of the emitted pulse is $1/3$ of the central frequency $f_0$, so the resolution in range is $c/B=3 \lambda$ while in cross-range it is $\lambda L/a= \lambda$.  We consider a high frequency microwave imaging regime with central frequency $f_0=60$GHz corresponding to $\lambda_0=5$mm. We make measurements for $N_f=21$ equally spaced frequencies spanning a bandwidth $B=20$GHz. The array aperture is $a=50$cm, and the distance from the array to the center of the IW is $L=50$cm. Then, the resolution 
is $\lambda_0 L/a=5$mm in the cross-range (direction parallel to the array) and $c_0/B=15$mm in range (direction of propagation). These parameters are typical in microwave scanning technology \cite{Laviada15}.

We consider an IW with $K=1681$ pixels which makes the dimension of $X=\bfrho\,\bfrho^*$ equal to $K^2= 2825761$. The pixel dimensions, i.e.,  the resolution of the imaging system, is $5 {\rm mm} \times  15 {\rm mm}$. The  total number of measurements is $N=N_r N_f=441$. Thus, we can form $N^2=194481$ cross-correlations over frequencies and locations.

Let us first note that with these values for $N$ and $K$, which in fact are not big, we cannot form the full matrix $({\bar \Ac} \otimes \Ac)$ so as to solve  (\ref{eq:vec1}) for $\mbox{vec}( X )$ because of its huge dimensions. For this reason, we propose to reduce the dimensionality of the problem to $K$ unknowns. Thus, we recover $\mbox{diag}(X)$ only, and neglect all the off-diagonal terms of $X$ corresponding to the interference terms $\rho_k \rho^*_{k'}$ for $k\ne k'$. We treat their contributions to the cross-correlated data as {\em noise}, which is absorbed in a fictitious vector $\bfeta$ using a {\em Noise Collector}. We stress that this {\em noise} is never small if $M\ge 2$, as its contribution to the the cross-correlated data is of order $O(M^2)$, while the contribution of $\mbox{diag}(X)$ is only of order $O(M)$.
\begin{figure}
\begin{center}
\includegraphics[scale=1]{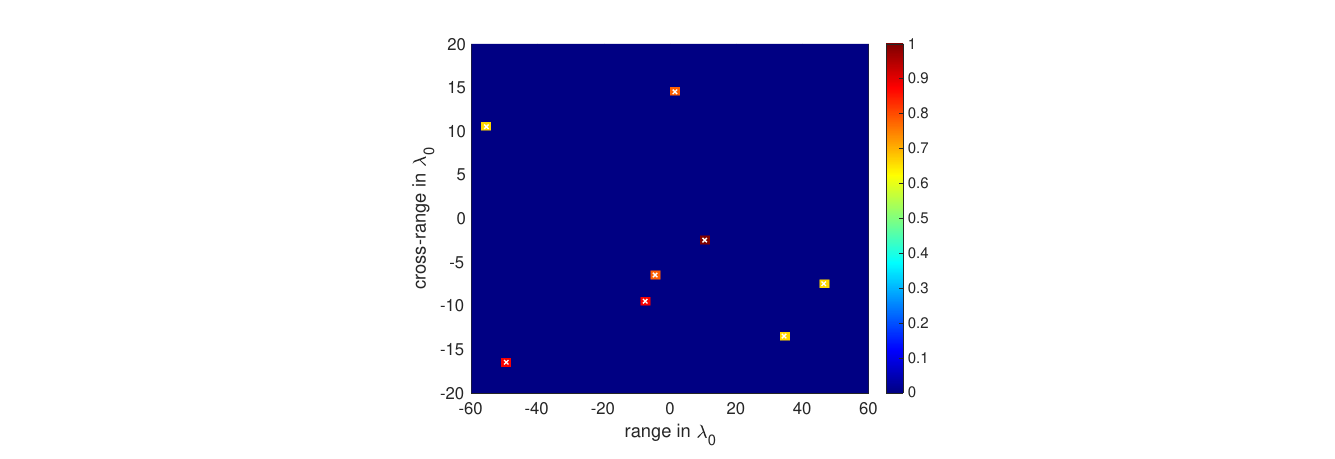}
 \end{center}
 \caption{The true ${\mbox{\boldmath$\chi$}}=\mbox{diag}(X)=\mbox{diag}(\bfrho\,\bfrho^*)$, i.e., the absolute values squared of the point sources amplitudes.. The dimension of the image is $K=1681$.} 
 \label{fig0} 
 \end{figure}

In the following examples, we consider imaging of $M=8$ point sources; see Fig. \ref{fig0}. 
Instead of the $N^2$  cross-correlated data which are, in principle, available, we only use $\mathcal{N}=21 N$ cross-correlated data picked at random. This reduces even more the dimensionality of the problem we solve.
In Fig. \ref{fig1}, we present the results when the used data is noise-free. The left column shows the results when we use the $\ell_1$ algorithm (\ref{eq:algo});
the top plot is the recovered image and the bottom plot the recovered ${\mbox{\boldmath$\chi$}}=\mbox{diag}(X)=\mbox{diag}(\bfrho\,\bfrho^*)$ vector.
The support of the sources is exact but the amplitudes are not. If it is important for an application to recover the amplitudes with  precision, one can consider in a second step the full problem (\ref{eq:vec1}) for $\mbox{vec}(X)$ with all the interference terms $\rho_k \rho^*_{k'}$ for $k\ne k'$,  but restricted to the exact  support found in the first step. If there is no noise in the data, this second step finds the exact values of the amplitudes efficiently using  an $\ell_2$ minimization method; see the right column of Fig. \ref{fig1}.

\begin{figure}
\begin{center}
\includegraphics[scale=1]{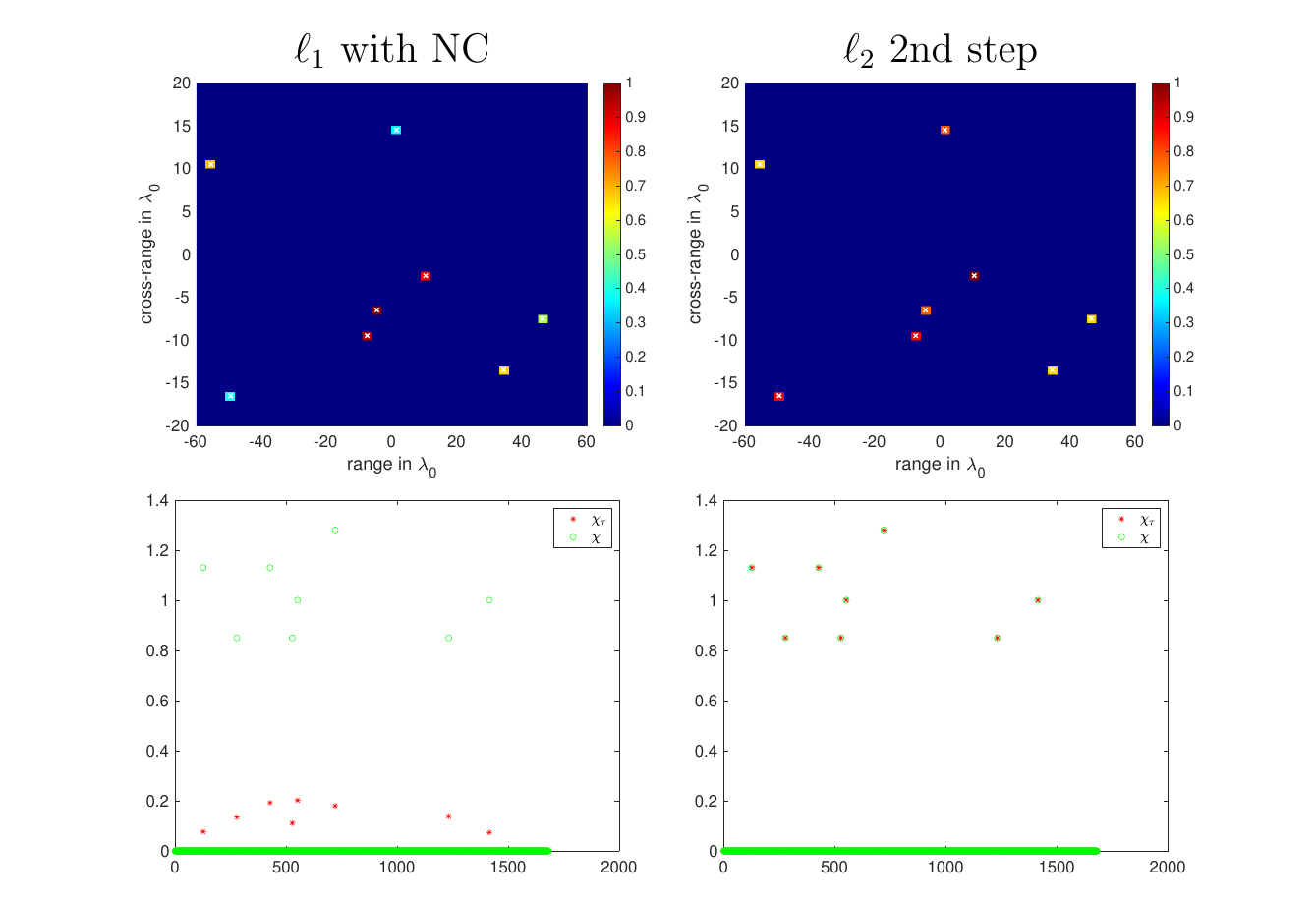}
    \end{center}
 \caption{Imaging $M=8$ sources using correlations and the NC. The dimension of the image is $K=1681$. The dimension of the linear data is $N=441$.  The $\ell_1$ images are obtained using  $21N$ of the $N^2$ correlation data. Noise free data.} 
 \label{fig1} 
 \end{figure}

In Figs. \ref{fig2} and \ref{fig3} we consider the same configuration of sources but we add white Gaussian noise to the data. The resulting SNR values are $10$dB and $0$dB, respectively.  In both cases, the solutions obtained in the first step look very similar to the one obtained in Fig. \ref{fig1} for noise free data. This is so, because the noise in the data is dominated by the neglected interference terms. The actual effect of the additive noise is only seen in the 2nd step when we solve for $\mbox{vec}(X)$, restricted to the support, using  an $\ell_2$ minimization method. Indeed, when the data are noisy  we cannot recover the exact values of the amplitudes. 
Still, since an $\ell_2$ method is used on the correct support, the reconstructions are extremely robust and give very good results, even when the SNR is $0$dB.

\begin{figure}
\begin{center}
\includegraphics[scale=1]{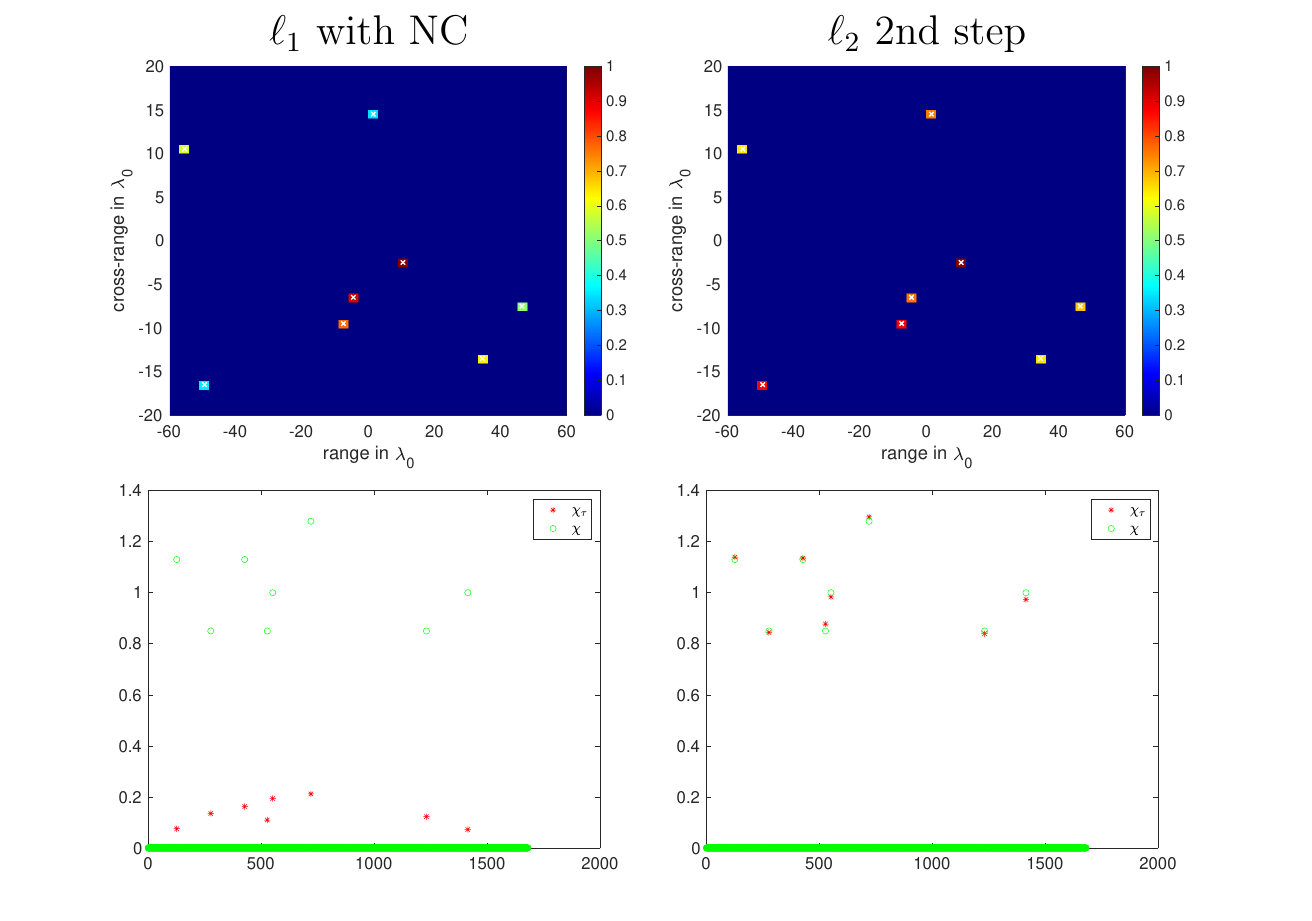}
    \end{center}
 \caption{Imaging $M=8$ sources using correlations and the NC. The dimension of the image is $K=1681$. The dimension of the linear data is $N=441$.  The $\ell_1$ images are obtained using  $21N$ of the $N^2$ correlation data. Data with 10dB SNR.} 
 \label{fig2} 
 \end{figure}
 
 
  \begin{figure}
\begin{center}
\includegraphics[scale=1]{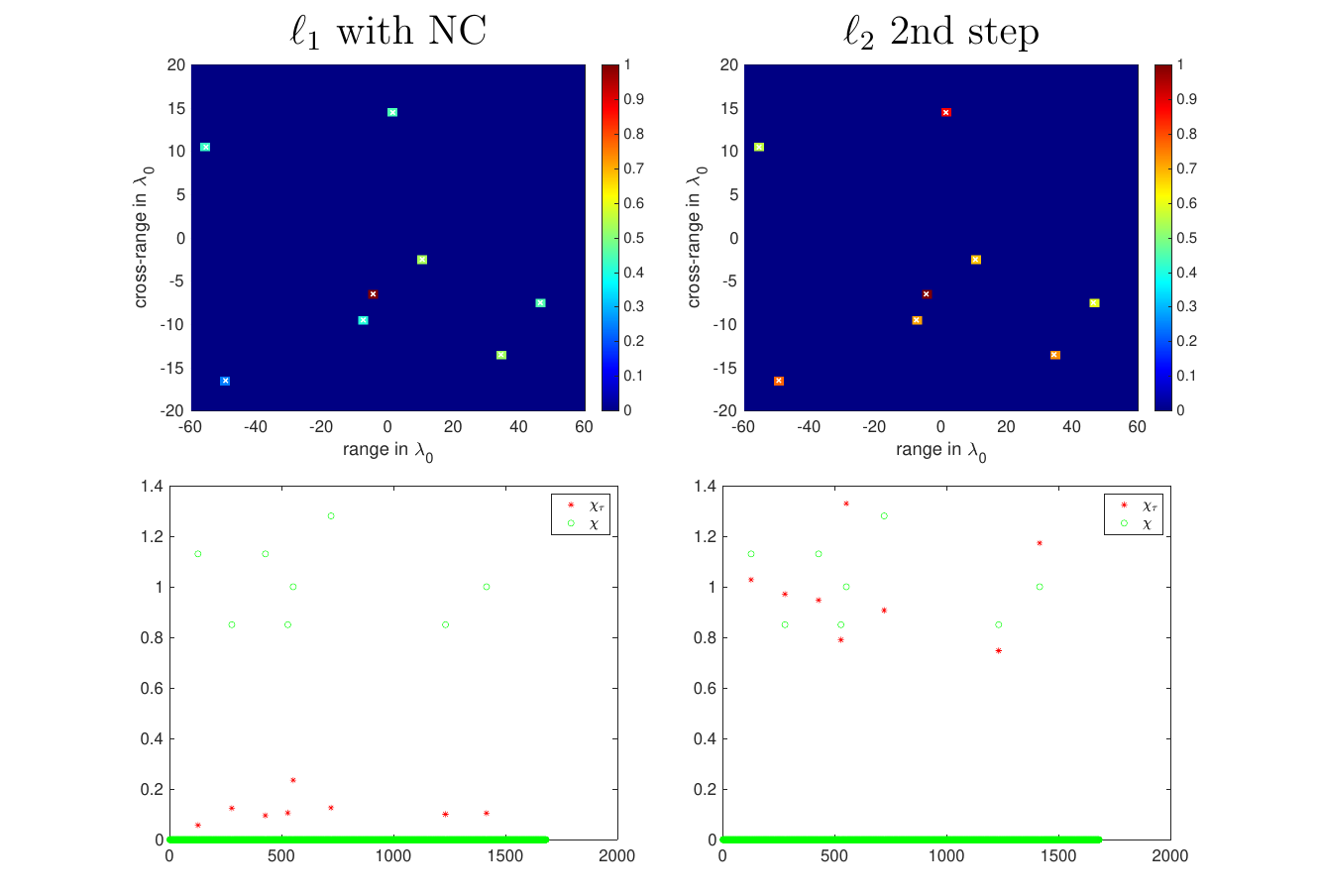}
    \end{center}
 \caption{Imaging $M=8$ sources using correlations and the NC. The dimension of the image is $K=1681$. The dimension of the linear data is $N=441$.  The $\ell_1$ images are obtained using  $21N$ of the $N^2$ correlation data. Data with 0dB SNR.} 
 \label{fig3} 
 \end{figure}

To illustrate the robustness of the reconstructions of the entire matrix $X=\bfrho\,\bfrho^*$ we also plot in Fig. \ref{fig4} the angle of $X_\tau$ compared to the angle of $X$ restricted on the support recovered during the first step. We get an exact reconstruction for noise-free data. The error in the reconstruction increases as the SNR decreases but the results are very satisfactory even for the $0$dB SNR case. 

  \begin{figure}
\begin{center}
\includegraphics[scale=1]{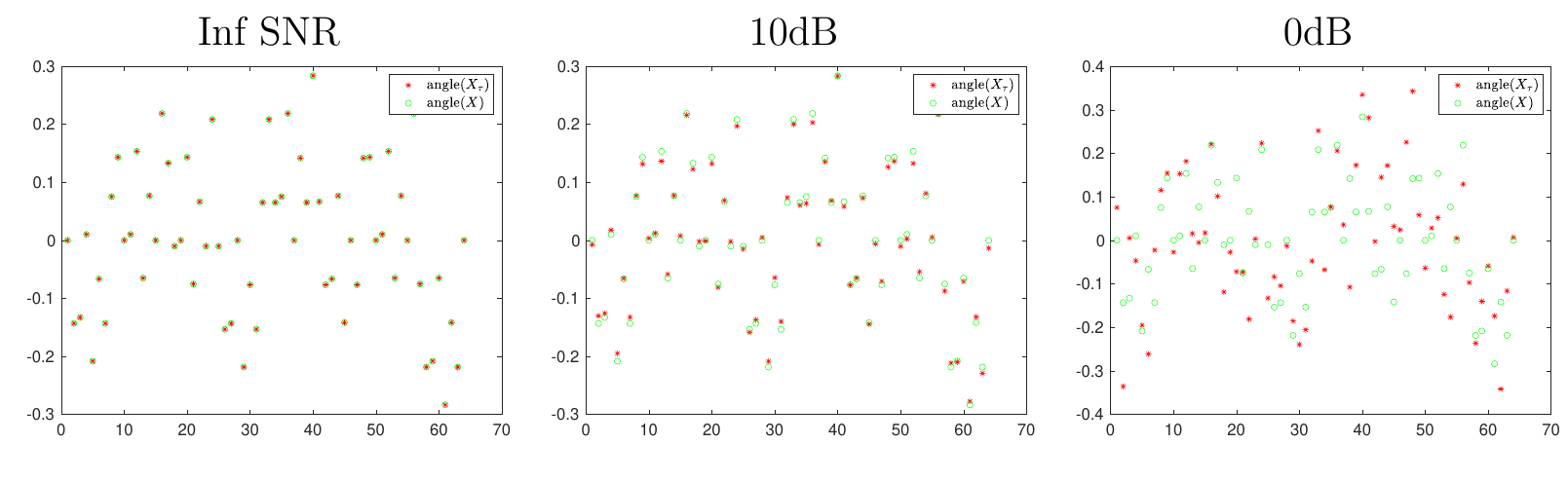}
    \end{center}
 \caption{Imaging $M=8$ sources using correlations and the NC. The dimension of the image is $K=1681$. The dimension of the linear data is $N=441$. The angle of the components of $X_\tau$ compared to angle of the components of the true $X$ restricted on the support.} 
 \label{fig4} 
 \end{figure}

Again, the big advantage of the proposed $\ell_1$ minimization  approach that seeks only for the components of $\mbox{diag}(X)$, and uses a {\em Noise Collector} to absorb the interference terms that are treated as noise, 
is that it is linear in the number of pixels $K$ instead of quadratic. This allows us to consider large scale problems. Moreover, as we observed in the results of Figs. \ref{fig1} to \ref{fig4}, the number of data $\mathcal{N}$ used to recover the images 
do not need to be $N^2$, but only a multiple of $N$.  

In Fig. \ref{fig_phase} we illustrate the performance of the proposed $\ell_1$ approach for different sparsity levels $M$ and data sizes $\mathcal{N}$. There is no additive noise added to the data in this figure. Success in recovering the true support of the unknown $\vect \chi$ corresponds to the value $1$ (yellow) and failure to $0$ (blue). The small phase transition zone (green) contains  intermediate values. The red line is the the estimate $\sqrt{\mathcal{N}}/(2\sqrt{\ln \mathcal{N}})$. These results are obtained by averaging over 10 realizations.
 
\begin{figure}[htbp]
\begin{center}
\includegraphics[scale=1]{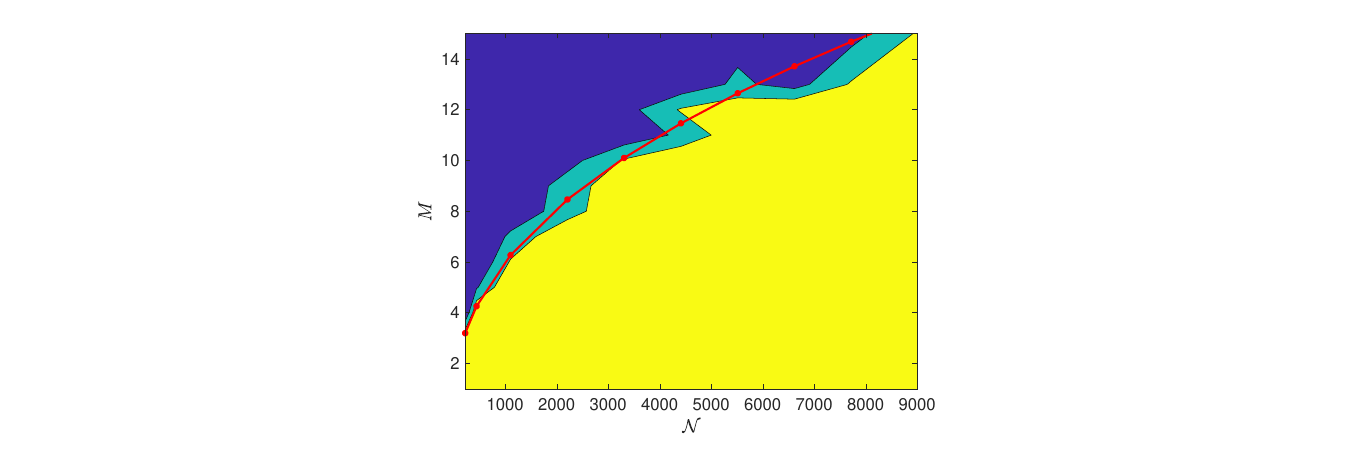}
\end{center}
\vspace*{-0.2cm}
\caption{Algorithm performance for exact support recovery during the first step using $\ell_1$ and the Noise Collector. Success corresponds to the value $1$ (yellow) and failure to $0$ (blue). The small phase transition zone (green) contains  intermediate values. The red line is the estimate $\sqrt{\mathcal{N}}/(2\sqrt{\ln \mathcal{N}})$. Ordinate and abscissa are the data used $\mathcal{N}$ and the sparsity $M$. }
\label{fig_phase}
\end{figure}

\section{Conclussions}
\label{sec:conclusions}

In this paper, we consider the problem of sparse signal recovery from cross correlation measurements. 
The unknown in this case is the correlated matrix signal $X=\vect \rho \vect \rho^*$ whose dimension grows quadratically with the size $K$ 
of $\vect \rho$ and, hence, inversion becomes computationally unfeasible as $K$ increases.
To overcome this issue, we propose a novel dimension reduction approach.  
Specifically, we vectorize the problem and consider as unknown only the diagonal terms $|\rho_i|^2$ of $X$ whose dimension is $K$ and are related to the data through a linear transformation. 
The off-diagonal interfernce terms $\rho_i \rho^*_j$ for $i\ne j$ are treated as noise and are absorbed using the {\em Noise Collector} approach introduced in \cite{Moscoso20b}.
This allows us to recover the signal exactly using efficient $\ell_1$-minimization algorithms. The cost of solving this dimension reduced problem is similar to the one using linear data. 
Furthermore, our numerical experiments show that the suggested approach is robust with respect to additive noise in the data. Finally, we point out that when using cross correlated 
data the maximum level of sparsity that can be recovered increases to $O(N/\sqrt{\ln N})$ instead of $O(\sqrt{N}/\sqrt{\ln N})$ for the linear data.


\section*{Acknowledgments}
The work of M. Moscoso was partially supported by Spanish MICINN grant FIS2016-77892-R. The work of A.Novikov was partially supported by NSF DMS-1813943 and AFOSR FA9550-20-1-0026. The work of G. Papanicolaou was partially supported by AFOSR FA9550-18-1-0519. The work of C. Tsogka was partially supported by AFOSR FA9550-17-1-0238 and FA9550-18-1-0519.



\appendix

\section{Proofs of the Theorems}
\label{sec:proofs}

\subsection{Proof of Theorem \ref{thm-ortho}}
\label{sec:proof2}

\begin{proof} To prove the first claim, we repeat the proof  of Theorem~2 from~\cite{Moscoso20b}.  Define $H_1$ as the convex hull of the columns of $\Cc$, and $H_2$ as the convex hull of the columns $\vect t_i$ of $T$ in the support of $\vect \chi$, as follows.
  \[
H_1 =H_1(\tau) =\left\{ x \in \mathbb{R}^{N}  \left| x=  \tau \sum_{i=1}^{\Sigma} \xi_i \vect c_i ,~\sum_{i=1}^{\Sigma} |\xi_i| \leq  1  \right. \right\}, 
\]
\[
H_2=\left\{ x \in \mathbb{R}^{N}  \left| x=   \sum_{ i \in \mbox{supp}(\vect \chi)} \xi_i \vect t_i ,~\sum_{i=1}^{K} |\xi_i| \leq  1  \right. \right\},  
\]
and
\[
H (\tau) =  \left\{ \xi h_1 + (1- \xi) h_2, 0 \leq \xi \leq 1, h_i \in H_i \right\}.
\]
Suppose the $(M+1)$-dimensional space $V$ is spanned by $\vect e$ and the column vectors $\vect t_j$, with $j$ in the support of $\vect \chi$. 
Denote by $\vect t_i^{v}$ the orthogonal projections of
$\vect t_i$ on $V$.
We will prove that $\mbox{supp}(\vect \chi_{\tau})  \subset \mbox{supp}(\vect \chi)$ if for any $\vect t_j$,  $j \not\in \mbox{supp}(\vect \chi)$, we have 
$\vect t^{v}_j \subset H(\tau)$ strictly (i.e. $\vect t^{v}_j  \cap \partial H(\tau)=\emptyset$)  
Fix $j \not\in \mbox{supp}(\vect \chi)$, and suppose

\begin{equation}\label{oye}
\vect t^{v}_j =  \xi_0 \vect t_0 + \sum_{k=1}^{M} \xi_k \vect t_{i_k}, \mbox{ where all }  i_k  \in  \mbox{supp}(\vect \chi), \vect t_0 =\frac{\vect e}{\| \vect e \|}.
\end{equation}
Suppose $|\xi_k|= \max_{n \leq M}  |\xi_n|$.  Multiply~(\ref{oye}) by $ |\xi_k|  \vect t^{v}_k/\xi_k$. Using~(\ref{no-random}) and~(\ref{no-random-2}) 
we obtain 
\[
c_0\frac{\sqrt{\ln {\cal N}}}{\sqrt{\cal N}} \geq |\xi_k| \left(1 - M  c_0 \frac{\sqrt{\ln {\cal N}}}{\sqrt{\cal N}} \right)
\]
Choose $\alpha$ in (\ref{eq:M}) so that 
\begin{equation}\label{alphachoice}
 M  c_0\frac{\sqrt{\ln {\cal N}}}{\sqrt{\cal N}} \leq \frac{1}{4}.
\end{equation}
Then,
\[
\left(1 - M  c_0\frac{\sqrt{\ln {\cal N}}}{\sqrt{\cal N}} \right) \geq \frac{3}{4},
\]
and therefore,
\[
  |\xi_k| \leq  \frac{4 c_0}{3}\frac{\sqrt{\ln {\cal N}}}{\sqrt{\cal N}}
\]
for all $k=0,1,2,\dots,M$. Hence, $\sum_{k=1}^{M}  |\xi_k| \leq 1/3$. 
By the Milman's extension of the Dvoretzky's theorem~\cite{milman} we can find $\tau_0=O(1)$   so that 
\[
4 c_0\frac{\sqrt{\ln {\cal N}}}{\sqrt{\cal N}}  \vect t_0 := \tilde{ \vect t}_0 \in H_1(\tau_0)
\]
with probability  $1-1/{\cal N}^{\kappa}$. Therefore,
\[
\!\!\!\!\!\!\!\!\!\!\!\!\!\!\!\!\!\!\!\!\!\!\!\!\!\!\!\!\!\!\!\!
\vect t^{v}_j =  \tilde{\xi}_0  \tilde{ \vect t}_0 + \sum_{k=1}^{M} \xi_k \vect t_{i_k}, \mbox{ where all }  i_k  \in  \mbox{supp}(\vect \chi) \,\mbox{ and } 
   |\tilde{\xi_0}| + \sum_{k=1}^{M}  |\xi_k| \leq  1/3   + 1/3 \leq  2/3  
\]
 and  $\tilde{ \vect t}_0 \in H_1(\tau)$
for all $\tau \geq \tau_0$. Therefore,  $\vect t^{v}_j \subset H(\tau)$ strictly.

To prove the second claim, we repeat the proof  of Theorem~3 from~\cite{Moscoso20b}.  Estimate~(\ref{alphachoice}) implies we can  assume  
$\langle \vect t_i, \vect t_j \rangle =0$ for $i \neq j$,  $i, j \in \mbox{supp}(\vect \chi)$ - this will only replace the constant $c_1$ in~(\ref{new-estimate2}) to $2 c_1 /\sqrt{3}$  at most.  
 Suppose $V^i$ are the $2$-dimensional spaces spanned by  $\vect e$ and $\vect t_i$ 
 for $i \in \mbox{supp}(\vect \chi)$.    By the Milman's extension of the Dvoretzky's theorem~\cite{milman}
   all $ \lambda H(\tau)  \cap V^i $ look like rounded rhombi depicted on Fig.~\ref{phi_map}, and  $ \lambda H_1(\tau) \cap V^i \subset B^i_{\lambda \tau}$ 
 with probability $1-N^{-\kappa}$, where $B^i_{ \lambda \tau}$  is a 2-dimensional $\ell_2$-ball  of radius  $ \lambda \tau c_0  \sqrt{\ln {\cal N}}/\sqrt{\cal N}$.  
  Thus $ \lambda H(\tau)  \cap V^i \subset  H^{i}_{\lambda \tau}$ with probability $1-N^{-\kappa}$,
  where   $H^{i}_{\lambda \tau }$ is the convex hull of $B^i_{\lambda \tau}$ and  a vector $\lambda  \vect f_i$, $\vect f_i= \chi_i \| \vect \chi \|^{-1}_{\ell_1}  \vect t_i$.  
  Then   
$\mbox{supp}(\vect \chi_\tau) = \mbox{supp}(\vect \chi)$, 
 if there exists $\lambda_0$ so that $\chi_i \vect t_i + \vect e$ lies on the flat boundary of  $H^{i}_{\lambda_0}$  for all $i \in \mbox{supp}(\vect \chi)$.

 \begin{figure}[htbp]
\begin{center}
\includegraphics[scale=1]{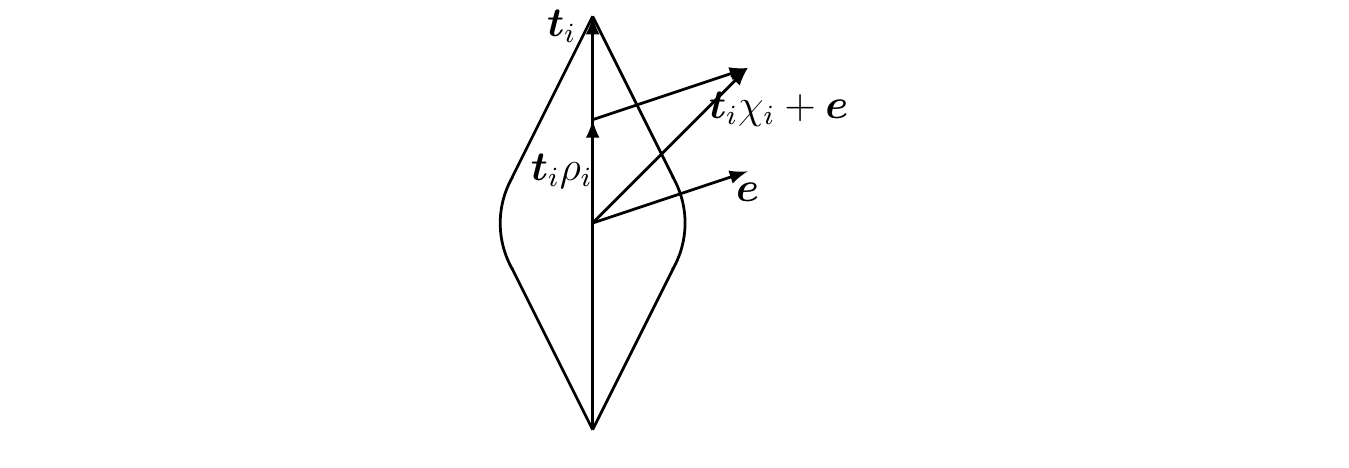}
\caption{ An intersection of $H(\tau)$ with the span$(\vect t_i, \vect e)$ is a rounded rhombus.
 }
\label{phi_map}
\end{center}
\end{figure}

 If $\min_{i \in \mbox{supp}(\vect \chi)} |\chi_i| \geq \gamma  \| \vect \chi \|_{\infty} $, then there exists a constant $c_2=c_2(\gamma)$ such that if 
$ \chi_i \vect t_i + \vect e$  lies on the flat boundary of  $H^{i}_{\lambda}$ for some $i$ and some $\lambda$, then there exists $\lambda_0$ so that
 $ \chi_i \vect t_i + c_2 \vect e$  lies on the flat boundary of  $H^{i}_{\lambda_0}$ for all $i \in \mbox{supp}(\vect \chi)$. 
   If 
 \begin{equation}\label{triangle}
 \frac{ | \langle  \vect t_i,  \chi_i \vect t_i + \vect e \rangle | }{  \| \chi_i \vect t_i + \vect e \|_{\ell_2} } \geq  \frac{ \tau c_0  \sqrt{\ln {\cal N}} }{\sqrt{\cal N} \| \vect f_i \|_{\ell_2}}=
 \frac{ \tau c_0 \| \vect \chi \|_{\ell_1} \sqrt{\ln {\cal N}} }{\sqrt{\cal N}  |\chi_i|},
 \end{equation}
 then $ \chi_i \vect t_i + c_2 \vect e$  lies on the flat boundary of  $H^{i}_{\lambda}$. 
 
 Since $| \langle  \vect t_i,  \chi_i \vect t_i + \vect e \rangle |  \geq  |\chi_i|/2$ by~(\ref{new-estimate1}), inequality~(\ref{triangle}) holds if 
 \[
  \frac{|\chi_i|}{  \| \chi_i \vect t_i + \vect e \|_{\ell_2} } \geq  \frac{ 2 \tau c_0 \| \vect \chi \|_{\ell_1} \sqrt{\ln {\cal N}} }{\sqrt{\cal N}  |\chi_i| }.
 \]   
By~(\ref{new-estimate2})  and using  $\| \vect \chi \|_{\ell_1}  \leq M $ the last inequality is true if
\[
M    \leq  \frac{\sqrt{\cal N}} {  \tau c_1 c_0 \sqrt{\ln {\cal N}} }. 
\]
The last inequality  is true if $\alpha$  in~(\ref{eq:M}) is small enough.
Thus,   $\mbox{supp}(\vect \chi_{\tau})  = \mbox{supp}(\vect \chi)$. 
\end{proof}

\subsection{Proof of Theorem~\ref{d-reduce}}
\label{sec:proof3}

\begin{proof} We need to verify that all conditions of Theorem~\ref{thm-ortho} are satisfied. Choose $c_0$, 
$\tau_0$, ${\cal N}_0$ and $\alpha$  so that Theorem~\ref{thm-ortho} is satisfied with probability $1-\frac{1}{3{\cal N}^{\kappa}}$. Note that we can increase $c_0$, 
$\tau_0$, ${\cal N}_0$ and decrease $\alpha$ in this proof if necessary.
We denote by
$({\bar \Ac} \otimes \Ac)_{k,l}$ the column of ${\bar \Ac} \otimes \Ac$ that arises from a tensor product $ \bar{ \vect a}_k \otimes \vect a_l$.
 If we use all $N^2$ of the data the columns ${\bar \Ac} \otimes \Ac$,
 then 
\[
\!\!\!\!\!\!\!\!\!\!\!\! \langle ({\bar \Ac} \otimes \Ac)_{k,l}, ({\bar \Ac} \otimes \Ac)_{m,n} \rangle =    \sum_{i=1}^{N} \sum_{j=1}^{N}  \bar{a}_{k,i} a_{l,j}  a_{m,i} \bar{a}_{n,j} = \langle \vect a_m, \vect a_k \rangle 
 \langle \vect a_l, \vect a_n \rangle. 
\]
In particular, all columns of   ${\bar \Ac} \otimes \Ac$ have length $1$. Therefore,
\[
\left| \langle \vect t_i,  \vect t_j \rangle \right| = \left| ({\bar \Ac} \otimes \Ac)_{i,i}, ({\bar \Ac} \otimes \Ac)_{j,j} \right| =   \left|  \langle \vect a_i, \vect a_j \rangle \right|^2 \leq \frac{\Delta^2}{N},
\]
and condition~(\ref{no-random-2}) is verified if we choose ${\cal N}_0$ large enough. 

Now we
obtain 
\begin{equation}\label{e-chi2}
\lambda_1 M \leq  \| \vect e\|_{\ell_2} \leq \lambda_2 M
\end{equation}
with high probability. Note that~(\ref{e-chi2}) implies~(\ref{new-estimate2}) because  $ \gamma M \leq \| \vect \chi \|_{\ell_1} \leq  M$.  We write
\[
\| \vect e\|^2_{\ell_2} = \| \vect \chi \|^2_{\ell_1} + 2 \| \vect \chi \|_{\ell_1} \Xi_1+ \Xi_2,
\]
where
\begin{equation}\label{xi1}
\Xi_1= \sum_{k, l, k \neq l} \bar{\rho}_k \rho_l  \langle \vect a_k, \vect a_l \rangle,
\end{equation}
and
\begin{equation}\label{xi2}
\Xi_2=  \sum_{\scriptsize \mbox{all indices different}}  \rho_k \bar{\rho}_l  \bar{\rho}_m \rho_n   \langle \vect a_m, \vect a_k \rangle 
 \langle \vect a_l, \vect a_n \rangle. 
\end{equation}

By Hanson-Wright inequality~(\ref{ourHW})
\[
\mathbb{P} \left( | \Xi_1| >t  \right) \leq 2 \exp \left(-   \frac{t^2/32}{\|{\bf M}\|^2_F} \right)
\]
where ${\bf M}$ is a matrix with components $ |\rho_k \rho_l|  \langle \vect a_k, \vect a_l \rangle$,
$\|{\bf M}\|_F$ is its Frobenius (Hilbert-Schmidt) norm. 
Since
$| \langle \vect a_k, \vect a_l \rangle| \leq \Delta/\sqrt{N}$,
we obtain   $\|{\bf M} \|_F \leq \Delta M/\sqrt{N}$ (in our set-up $\| \vect \rho \|_{\ell_\infty}=1$).
Take $t= \gamma M/8 \leq   \| \vect \chi \|_{\ell_1}/8 $
and obtain  
\[
\!\!\!\!\!\!\!\!\!\!\!\! \mathbb{P} \left( | \Xi_1| > \gamma M/8  \right) \leq 2 \exp \left(- c  N \right),  c=c(\gamma),
\]
which is negligible for large $N$. Thus
\begin{equation}\label{ono}
 | \Xi_1|  \leq \frac{\|\vect \chi  \|_{\ell_1}}{8} 
\end{equation}
with probability $1- 2 \exp  \left(- c  N \right)$.
Observe that 
$\Xi_2= (\Xi_1)^2- \Xi_3$, where
\[
\Xi_3=  \sum_{\scriptsize m =l \mbox{ or } k=n \mbox{ or both}}  \rho_k \bar{\rho}_l  \bar{\rho}_m \rho_n   \langle \vect a_m, \vect a_k \rangle 
 \langle \vect a_l, \vect a_n \rangle. 
\] 
For $\Xi_3$ we can use a deterministic estimate:
\[
|\Xi_3| \leq 2 \frac{c^2_0 M^3}{N} \leq c \alpha \frac{\|\vect \chi  \|^2_{\ell_1}}{\sqrt{\ln N}} \leq  \frac{\|\vect \chi  \|^2_{\ell_1}}{16} .
\]
For $(\Xi_1)^2$ we use~(\ref{ono}). Using the union bound, we obtain
\begin{equation}\label{e-chi}
\frac{1}{2} \|\vect \chi  \|^2_{\ell_1} \leq \| \vect e\|^2_{\ell_2}  \leq \frac{3}{2} \|\vect \chi  \|^2_{\ell_1} 
\end{equation}
with probability  $1- 2 \exp\left(- c   N \right)$. Thus, (\ref{e-chi2}) holds
with probability  $1- 2 \exp\left(- c   N \right)$.

We will now prove~(\ref{no-random}).  For  $m \not\in \mbox{supp}(\vect \chi)$,  consider a random variable
\begin{equation}\label{e_z}
\!\!\!\!\!\!\!\!\!\!\!\! \!\!\!\!\!\!\!\!\!\!\!\!   \Theta_m =\langle \vect t_m,  \vect e \rangle   
= \sum_{k, l, k \neq l} \bar{\rho}_k \rho_l \langle \vect t_m,  ({\bar \Ac} \otimes \Ac)_{k,l} \rangle 
= \sum_{k, l, k \neq l} \bar{\rho}_k \rho_l  \langle \vect a_m, \vect a_k \rangle  \langle \vect a_m, \vect a_l \rangle .
\end{equation}
We have
\begin{equation}\label{easycase}
 \left| \langle \vect t_m,  ({\bar \Ac} \otimes \Ac)_{k,l} \rangle \right| =    \left|  \langle \vect a_m, \vect a_k \rangle  \langle \vect a_m, \vect a_l \rangle \right| \leq \frac{\Delta^2}{N}
\end{equation}
if $m \neq k$, and $m \neq l$. If ${\bf M}$ is a matrix with components $ |\rho_k \rho_l|  \langle \vect a_m, \vect a_k \rangle  \langle \vect a_m, \vect a_l \rangle$, 
then  $\|{\bf M} \|_F \leq \Delta^2 M/N$. 
Using~(\ref{e-chi2})  choose $t = c_0 \frac{\sqrt{\ln N^2}}{N} \| \vect e \|_{\ell_2} > c_0 \frac{\gamma}{2} \frac{M  \sqrt{\ln N}}{N} $ in Hanson-Wright inequality~(\ref{ourHW})
to obtain: 
\[
\!\!\!\!\!\!\!\!\!\!\!\!\!\!\!\!\! \!\!\!\!\!\!\!\!\!\!\!\!   \mathbb{P} \left( | \Theta_m| >  c_0 \frac{\sqrt{\ln N^2}}{N} \| \vect e \|_{\ell_2} 
\right) \leq \mathbb{P} \left( | \Theta_m| >  c_0 \frac{\gamma}{2} \frac{M  \sqrt{\ln N}}{N} 
\right) \leq 
2 \exp \left(- \frac{ \gamma^2 c^2_0 \ln N}{128 \Delta^4}\right). 
\]
Then~(\ref{no-random}) holds with probability $1-\frac{1}{3{\cal N}^{\kappa}}$ if $c_0$ is large enough.

We will now prove~(\ref{new-estimate1}). For  $m \in \mbox{supp}(\vect \chi)$  decompose
\[
\!\!\!\!\!\!\!\!\!\!\!\! \!\!\!\!\!\!\!\!\!\!\!\!   \Theta_m =\langle \vect t_m,  \vect e \rangle   =  \Theta^1_m+ \Theta^2_m
\]
where
\[
 \Theta^1_m= \sum_{k, l, k \neq l, k \neq m, l \neq m } \bar{\rho}_k \rho_l  \langle \vect a_m, \vect a_k \rangle  \langle \vect a_m, \vect a_l \rangle
\]
and
\begin{equation}\label{bad_guy}
 \Theta^2_m= \sum_{k,  k \neq m  } \left( \bar{\rho}_m \rho_k +   \bar{\rho}_k \rho_m \right)    \langle \vect a_m, \vect a_k \rangle
\end{equation}
The distribution of the random variable $ \Theta^1_m$ has exactly the same behavior as $\Theta_m$ for $m \not\in \mbox{supp}(\vect \chi)$. We therefore have
\[
\!\!\!\!\!\!\!\!\!\!\!\! \!\!\!\!\!\!\!\!\!\!\!\!  
\mathbb{P} \left( | \Theta^1_m| >  \frac{\gamma}{4}   \| \chi \|_{\ell_\infty}  \right) \leq 2 \exp\left(-  c \frac{N^2}{\Delta^4 M^2 }\right) \leq  2 \exp(- \tilde{c} \ln N/\alpha^2)
 \leq  \frac{1}{6}\frac{1}{{\cal N}^\kappa},  
\]
  by Hanson-Wright inequality~\ref{ourHW} if $\alpha$ is small enough.  .
 If $m=l$ (or   $m=k$) then
\begin{equation}\label{too_weak}
\left| \langle \vect t_m,  ({\bar \Ac} \otimes \Ac)_{k,m} \rangle \right| =    \left|  \langle \vect a_m, \vect a_k \rangle  \right| \leq \frac{\Delta}{\sqrt{N}}.
\end{equation}
If we condition on $\vect \rho_m$, then  $\Theta^2_m$ is a sum of independent random variables. 
Therefore  by  Hoeffding's inequality 
\[
\mathbb{P}\left( |\Theta^2_m| > t \right) \leq 2 \exp\left( -c \frac{t^2}{ b^2} \right),
\mbox{ where } 
b^2 \leq  \frac{c^2_0 M}{N} \leq \frac{\Delta^2 \alpha}{\ln N}.
\]  
Choosing $t$ appropriately we obtain 
\[
\mathbb{P} \left( | \Theta^2_m| > \frac{\gamma}{4} \| \chi \|_{\ell_\infty} \right) \leq   \frac{1}{6} \frac{1}{{\cal N}^\kappa}.
\]
by choosing $\alpha$ small enough. Using the union bound we conclude that
 \[
\mathbb{P} \left( | \Theta_m| > \frac{\gamma}{2} \| \chi \|_{\ell_\infty} \right) \leq   \frac{1}{3}  \frac{1}{{\cal N}^\kappa}
\] 
for $m \in \mbox{supp}(\vect \chi)$.
Applying the union bound we conclude that conditions~(\ref{new-estimate1}),~(\ref{no-random}) and estimates in Theorem~\ref{thm-ortho} hold with probability $1-\frac{1}{{\cal N}^{\kappa}}$.
This completes the proof.
\end{proof}
\begin{remark}\label{new_e}
 The proof of Theorem~\ref{d-reduce} reveals why we had to assume~(\ref{new-estimate1}) for  $m \in \mbox{supp}(\vect \chi)$. When   $m \not\in \mbox{supp}(\vect \chi)$ then $\langle \vect t_m,  \vect e \rangle$ is
 estimated in~(\ref{e_z}) using~(\ref{easycase}). When   $m \in \mbox{supp}(\vect \chi)$ then $\langle \vect t_m,  \vect e \rangle$ contains  $ \Theta^2_m$ given by~(\ref{bad_guy}). For  $ \Theta^2_m$ we cannot use~(\ref{easycase}),
and  we have to use a weaker 
 estimate~(\ref{too_weak}).
  \end{remark}

\subsection{A deterministic version of Theorem \ref{d-reduce}}
\label{sec:proof4}

\begin{theorem}\label{d-reduce-norandom} 
Suppose $X$ is a solution of~(\ref{eq:vec1}), $ \vect \chi =\mbox{diag}(X)$ is $M$-sparse, $\vect{d} \in \mathbb{C}^{\cal N}$, ${\cal N} =N^2$, and
$T=({\bar \Ac} \otimes \Ac)_{\vect \chi}: \mathbb{C}^{\cal K} \to \mathbb{C}^{\cal N}$.
Fix $\beta>1$,  and draw  $\Sigma={\cal N}^{\beta}$  columns for 
$\Cc$, independently, from the  uniform distribution on 
$\mathbb{S}^{{\cal N}-1}$ and define $\gamma$ as in (\ref{def:gamma}) and $\Delta$ as in (\ref{def:Delta}).
Then, for any $\kappa > 0$, there are constants $\alpha=\alpha(\kappa, \gamma, \Delta)$, $\tau=\tau(\kappa, \beta)$, and ${\cal N}_0= {\cal N}_0(\kappa, \beta, \gamma, \Delta)$ such that
the following holds.
  If  
$M \leq \alpha \sqrt{N}$ 
and $ \vect \chi_{\tau}$  is the solution~(\ref{rho_tt}),
 then
$ \mbox{supp}(\vect \chi) = \mbox{supp}(\vect \chi_{\tau})$ for  all ${\cal N}>{\cal N}_0$ 
with  probability  $1-1/{\cal N}^{\kappa}$. 
\end{theorem}

\begin{proof}  We need to verify that all conditions of Theorem~\ref{thm-ortho} are satisfied non-probabilistically.
Conditions~(\ref{no-random-2}) is already verified in the proof of Theorem~\ref{d-reduce} under even weaker assumptions than in 
Theorem~\ref{d-reduce-norandom}. Therefore we only need to verify 
estimates~(\ref{new-estimate2}),~(\ref{new-estimate1}) and~(\ref{no-random}).

Since
\[
\vect e =\sum_{k \neq l} \bar{\rho}_k \rho_l ({\bar \Ac} \otimes \Ac)_{k,l}, 
\]
we have
\[
\!\!\!\!\!\!\!\!\!\!\!\!\!\!\!\!\!\!\!\!\!\!\!\!\!\!\!\!\!\!\!\!\!\!\!\!  \| \vect e\|^2_{\ell_2} \leq  \frac{2 \Delta}{\sqrt{N}} \sum_{\scriptsize{all~indices}} \chi_k \left|\rho_{m_1} \right| \left|\rho_{m_2} \right|
+ \frac{\Delta^2}{N} \sum_{\scriptsize{all~indices}}  \left|\rho_{m_1} \right| \left|\rho_{m_2} \right| \left|\rho_{k_1} \right| \left|\rho_{k_2} \right| 
+  \sum_{k, m} \chi_k \chi_m   
\] 
\[
  \!\!\!\!\!\!\!\!\!\!\!\! \!\!\!\!\!\!\!\!\!\!\!\! \!\!\!\!\!\!\!\!\!\!\!\! \leq  2 \Delta \alpha  \| \vect \chi \|_{\ell_1}\| \vect \rho \|^2_{\ell_2}
+ \Delta^2 \alpha^2   \| \vect \rho \|^4_{\ell_2} + \| \vect \chi \|^2_{\ell_1}  =  (  1+ \Delta \alpha )^2  \| \vect \chi \|^2_{\ell_1}.
\]
Thus estimate~(\ref{new-estimate2}) holds.

A non-probabilistic version of estimate~(\ref{new-estimate1}) is as follows.  For $m \in \mbox{supp}(\vect \chi)$ we have 
\[
\!\!\!\!\!\!\!\!\!\!\!\! \!\!\!\!\!\!\!\!\!\!\!\!  \left| \langle \vect t_m,  \vect e \rangle \right|   = \left|  \sum_{k, l, k \neq l} \bar{\rho}_k \rho_l \langle \vect t_m,  ({\bar \Ac} \otimes \Ac)_{k,l} \rangle \right| 
\leq 2  \sum_{k, k \neq m}  \left| \rho_k \right|  \left| \rho_m \right| \left| \langle \vect t_m,  ({\bar \Ac} \otimes \Ac)_{k,m} \rangle \right| 
\]
\[
\!\!\!\!\!\!\!\!\!\!\!\! \!\!\!\!\!\!\!\!\!\!\!\!  +  \sum_{k,l, k \neq l \neq m}  \left| \rho_k \right|  \left| \rho_l \right|  \left| \langle \vect t_m,  ({\bar \Ac} \otimes \Ac)_{k,l} \rangle \right| \leq 
\frac{2 \Delta}{\sqrt{N}}   \sum_{k}  \left| \rho_k \right|  \left| \rho_m \right|  + \frac{\Delta^2}{N}   \sum_{k, l}  \left| \rho_k \right|  \left| \rho_l \right| 
\]
\[
\!\!\!\!\!\!\!\!\!\!\!\! \!\!\!\!\!\!\!\!\!\!\!\!  \leq \left(  \frac{2 \Delta M}{\sqrt{N}} +  \frac{\Delta^2 M^2}{N} \right) \| \vect \rho \|^2_{\ell_\infty}
= \left(  \frac{2 \Delta M}{\sqrt{N}} +  \frac{\Delta^2 M^2}{N} \right) \| \vect \chi \|_{\ell_\infty}
 \leq \frac{\gamma}{2} \| \vect \chi \|_{\ell_\infty}.
\]
if $\alpha$ is small enough.

We now obtain a lower bound on $\| \vect e \|_{\ell_2}$. For $\Xi_1$ and $\Xi_2$ in~(\ref{xi1}) and~(\ref{xi2}), respectively, we have
\[
| \Xi_1| \leq \frac{\Delta M^2 }{\sqrt{N}} \leq \Delta \alpha M,~~| \Xi_2 | \leq \frac{\Delta^2 M^4 }{N} \leq \Delta^2 \alpha^2 M^2. 
\]  
Since
\[
\| \vect e\|^2_{\ell_2} = \| \vect \chi \|^2_{\ell_1} + 2 \| \vect \chi \|_{\ell_1} \Xi_1+ \Xi_2, \mbox{ and }  \| \vect \chi \|_{\ell_1}=M
\]
we can choose $\alpha$ so that
\[
M/2 = \| \vect \chi \|_{\ell_1}/2 \leq \| \vect e\|^2_{\ell_2}. 
\]
To show~(\ref{no-random}) observe that
\[
 \left| \langle \vect t_m,  ({\bar \Ac} \otimes \Ac)_{k,l} \rangle \right| =    \left|  \langle \vect a_m, \vect a_k \rangle  \langle \vect a_m, \vect a_l \rangle \right| \leq \Delta^2/N,
\]
because $m \neq k$, and $m \neq l$. Therefore
\[
\!\!\!\!\!\!\!\!\!\!\!\! \!\!\!\!\!\!\!\!\!\!\!\ \left| \langle \vect t_m,  \vect e \rangle   \right|
= \left| \sum_{k, l, k \neq l} \bar{\rho}_k \rho_l \langle \vect t_m,  ({\bar \Ac} \otimes \Ac)_{k,l} \rangle \right|
\leq \frac{\Delta^2 M^2}{N} \leq \frac{\Delta^2 \alpha \| \vect e\|_{\ell_2}}{\sqrt{N}},  
\]
and~(\ref{no-random})  follows either for choosing $\alpha$ small or $\ln N$ large.
\end{proof}

\subsection{Hansen-Wright's Inequality for bounded symmetric random variables}
\label{subsec:HW}

For simplicity of presentation
 all random variables here are real.   Suppose $X_i$ are independent sub-gaussian random variables, $\mathbb{E}(X_i)=0$, and the sub-gaussian norms $\| X_i \|_{\psi_2} \leq K$. Consider
\[
\Xi= \sum_{i,j} X_j X_i m_{ij},
\]
where $m_{ij}$ are entries of a deterministic $M \times M$ diagonal-free (i.e. $m_{ii}=0$) matrix ${\bf M}$. 
The Hanson-Wright inequality (see e.g.~\cite{RudVersh}) is
\begin{equation}\label{classical}
\mathbb{P} \left( | \Xi| >t  \right) \leq 2 \exp \left(- c \min \left(  \frac{t^2}{K^4 \|{\bf M}\|^2_F},  \frac{t}{K^2 \|{\bf M}\|}\right) \right)
\end{equation}
where  $\|{\bf M}\|_F$ is the Frobenius (Hilbert-Schmidt) norm of ${\bf M}$, and $\|{\bf M}\|$ is its operator norm. If we use this inequality in the proof of our Theorem~\ref{d-reduce}, then 
the result becomes weaker than Theorem~\ref{pnas} by a factor of $\sqrt{\ln N}$ because in our setting 
\[
 \min \left(  \frac{t^2}{K^4 \|{\bf M}\|^2_F},  \frac{t}{K^2 \|{\bf M}\|}\right) = \frac{t}{K^2 \|{\bf M}\|}.
\]
In order to obtain Theorem~\ref{d-reduce} in its present form, we need a slight strengthening of~(\ref{classical}). Our proof is a modification of two proofs 
from~\cite{vershynin} and~\cite{RudVersh}. It  may already exist in the literature, but we were not able to find it. Therefore we provide it here for the reader's convenience. We assume that our random variables are symmetric and bounded. 
This holds if a random variable is uniformly distributed on the (complex) unit circle as in Theorem~\ref{d-reduce}.  

\begin{theorem}\label{hw-new}(Hansen-Wright inequality for bounded symmetric random variables)
Suppose $X_i$ are independent symmetric random variables, with $\| X_i \|_{\ell_\infty} \leq K$.  Let
$\Xi= \sum_{i \neq j} X_j X_i m_{ij}$. Then
\begin{equation}\label{ourHW}
\mathbb{P} \left( | \Xi| >t  \right) \leq 2 \exp \left(-  \frac{t^2/32}{K^4 \|{\bf M}\|^2_{F}} \right).
\end{equation}
\end{theorem}

\begin{proof} 
By replacing  $X_i$ with $X_i/K$  we can assume  $K=1$.  By Chebyshev's inequality 
\begin{equation}\label{unos}
\mathbb{P} \left(  \Xi >t  \right) = \mathbb{P} \left(  e^{ \lambda \Xi} > e^{\lambda t}  \right) \leq e^{-\lambda t} \mathbb{E} \left( e^{\lambda \Xi} \right)
\end{equation}
for any $\lambda >0$. We now use decoupling.   Consider independent Bernoulli random variables $\mu_i = 0$ or $1$ with probability  $1/2$. Since
$\mathbb{E}(\mu_i)(1- \mu_j) =1/4$ for $i \neq j$ we conclude $\Xi = 4 \mathbb{E}_{\mu} \Xi_\mu$, where
\[
\Xi_{\mu}=  \sum_{i \neq j} \mu_j (1- \mu_i) X_j X_i m_{ij},
\]
and $\mathbb{E}_{\mu}$ is conditional expectation with respect to $\mu=(\mu_1, \dots, \mu_M)$. Using independence of $X=(X_1, \dots , X_M)$ and $\mu$, and applying  Jensen's inequality
we obtain
\[
\mathbb{E} \left( e^{\lambda \Xi} \right)=\mathbb{E}_{X} \left( e^{\lambda \Xi} \right) \leq \mathbb{E}_{X} \mathbb{E}_{\mu} \left( e^{4 \lambda \Xi_{\mu}} \right)
= \mathbb{E}_{\mu}  \mathbb{E}_{X} \left( e^{4 \lambda \Xi_{\mu}} \right)
\]
where $\mathbb{E}_{X}$ is conditional expectation with respect to $X$. This implies there exist a realization of $\mu$ such that
$\mathbb{E}_{X} \left( e^{\lambda \Xi} \right) \leq \mathbb{E}_{X} \left( e^{4 \lambda \Xi_{\mu}} \right)$
for this $\mu$. Fix this $\mu$ and the corresponding  set of indices $\Lambda_{\mu}=\{i | \delta_i =1\}$. Then we can write
$\Xi_{\mu} = \sum_{i \neq j, i \in \Lambda_{\mu}, j \in  \Lambda^c_{\mu}} X_j X_i m_{ij}$.  Since  the random variables $X_i$, $i \in \Lambda_{\mu}$
and  $X_i$, $i \in \Lambda^c_{\mu}$ are independent,  their distribution will not change if we replace  $X_i$, $i \in \Lambda^c_{\mu}$ by 
$X^{'}_i$, $i \in \Lambda^c_{\mu}$, where $X^{'}_i$ is an independent copy of $X_i$. In other words we have
\[
\mathbb{E}_{X, X'} \left( e^{\lambda \Xi} \right) \leq  \mathbb{E}_{X} \left( e^{4 \lambda \tilde{\Xi}_{\mu}} \right),
\mbox{ where }
\tilde{\Xi}_{\mu} =  \sum_{i \neq j, i \in \Lambda_{\mu}, j \in  \Lambda^c_{\mu}} X^{'}_j X_i m_{ij}. 
\]
We now claim that
\[
\mathbb{E}_{X, X'} \left( e^{4 \lambda \tilde{\Xi}_{\mu}} \right) \leq  \mathbb{E}_{X, X'} \left( e^{4 \lambda \tilde{\Xi}} \right), 
\mbox{ where }
\tilde{\Xi} =  \sum_{i \neq j} X^{'}_j X_i m_{ij}. 
\]
Indeed,  Lemma 6.1.2  in~\cite{vershynin} states that
$\mathbb{E}  \left(F(Y) \right) \leq   \mathbb{E}  \left(F(Y+Z) \right)$ for any convex function $F$, if $Y$ and $Z$ are independent and $\mathbb{E}(Z)=0$. In our case
we take $F(x)=e^{4 \lambda x}$,  $Y= \tilde{\Xi}_{\mu}$ and $Z= \tilde{\Xi} - \tilde{\Xi}_{\mu}$.  If we condition on $X_i$, $i \in \Lambda_{\mu}$ and  $X^{'}_i$, $i \in \Lambda^c_{\mu}$, 
then $Y$ is fixed,  $Z$ is independent $Y$ and its conditional expectation is zero. Hence the following decoupling estimate is obtained.
\begin{equation}\label{decouple}
\mathbb{E} \left( e^{\lambda \Xi} \right)  \leq  \mathbb{E} \left( e^{4 \lambda \tilde{\Xi}} \right), 
\mbox{ where }
\tilde{\Xi} =  \sum_{i \neq j} X^{'}_j X_i m_{ij}. 
\end{equation}
By independence
\[
 \mathbb{E} \left( e^{4 \lambda \tilde{\Xi}} \right) = \prod_{i \neq j} \mathbb{E} \left( e^{4 \lambda m_{ij} X_i X^{'}_j} \right). 
\]
Since random variables are symmetric
\[
\!\!\!\!\!\!\!\!\!\!\!\!\!\!\!\!\!\!\!\!\!\!\!\!\!\!\!\!\!\!\!\!\!
 \mathbb{E} \left( e^{4 \lambda m_{ij} X_i X^{'}_j} \right) = \frac{1}{2}   \mathbb{E} \left( e^{4 \lambda m_{ij} \left| X_i X^{'}_j\right| }  + e^{- 4 \lambda m_{ij} \left| X_i X^{'}_j\right|} \right)
\leq    \mathbb{E} \left(  e^{8 \lambda^2 m^2_{ij} \left| X_i X^{'}_j\right|^2 }   \right) \leq e^{8  \lambda^2 m^2_{ij}}.
\]
Using the last two estimates in~(\ref{decouple}) we obtain
\[
\mathbb{E} \left( e^{\lambda \Xi} \right)  \leq    e^{8  \lambda^2 \|{\bf M} \|_F^2}.
\]
Plugging the lsat inequality in~\ref{unos} and optimizing over $\lambda$ we obtain
\[
\mathbb{P} \left(  \Xi >t  \right)  \leq \inf_{\lambda >0} e^{-\lambda t + 8  \lambda^2  \|{\bf M} \|_F^2} < e^{- \frac{t^2}{32 \|{\bf M} \|_F^2}}.
\]
\end{proof}

 \end{document}